\documentclass[mnsc]{informs3} 

\renewcommand{\vec}[1]{{\mathbf #1}}
\newcommand{\perfmap}{M}

\usepackage[ruled]{algorithm2e}
\usepackage{cases}
\usepackage{authblk}
\usepackage{enumitem,dsfont}

\newcommand{\reals}{\mathbf{R}}
\newcommand{\E}{\mathbf{E}}
\renewcommand{\Pr}{\mathbf{Pr}}

\newcommand{\OPT}{\mathrm{OPT}}
\newcommand{\clow}{p}
\newcommand{\chigh}{q}

\usepackage{natbib}
\bibpunct[, ]{(}{)}{,}{a}{}{,}%
\def\bibfont{\small}%
\usepackage{natbib}
 \bibpunct[, ]{(}{)}{,}{a}{}{,}%
 \def\bibfont{\small}%

\usepackage{amssymb}

\TheoremsNumberedThrough     
\ECRepeatTheorems

\EquationsNumberedThrough    

\MANUSCRIPTNO{MS-0001-1922.65}

\begin{document}


\RUNAUTHOR{Sekar, Vojnovic, and Yun}

\RUNTITLE{A Test Score Based Approach to Stochastic Submodular Optimization}

\TITLE{A Test Score Based Approach to Stochastic Submodular Optimization}


\ARTICLEAUTHORS{%
\AUTHOR{Shreyas Sekar}
\AFF{Harvard Business School, Boston, MA,
\EMAIL{ssekar@hbs.edu}} 
\AUTHOR{Milan Vojnovic}
\AFF{Department of Statistics, London School of Economics (LSE), London, UK,
\EMAIL{m.vojnovic@lse.ac.uk}}
\AUTHOR{Se-Young Yun}
\AFF{Department of Industrial and System Engineering, KAIST, South Korea,
\EMAIL{yunseyoung@kaist.ac.kr}}
} 

\ABSTRACT{We study the canonical problem of maximizing a stochastic submodular function subject to a cardinality constraint, where the goal is to select a subset from a ground set of items with uncertain individual performances to maximize their expected group value. Although near-optimal algorithms have been proposed for this problem, practical concerns regarding scalability, compatibility with distributed implementation, and expensive oracle queries persist in large-scale applications. Motivated by online platforms that rely on individual item scores for content recommendation and team selection, we propose a special class of algorithms that select items based solely on individual performance measures known as test scores. The central contribution of this work is a novel and systematic framework for designing test score based algorithms for a broad class of naturally occurring utility functions. We introduce a new scoring mechanism that we refer to as \emph{replication test scores} and prove that as long as the objective function satisfies a diminishing returns property, one can leverage these scores to compute solutions that are within a constant factor of the optimum.  We then extend our results to the more general stochastic submodular welfare maximization problem, where the goal is to select items and assign them to multiple groups to maximize the sum of the expected group values. For this more difficult problem, we show that replication test scores can be used to develop an algorithm that approximates the optimum solution up to a logarithmic factor. The techniques presented in this work bridge the gap between the rigorous theoretical work on submodular optimization and simple, scalable heuristics that are useful in certain domains. In particular, our results establish that in many applications involving the selection and assignment of items, one can design algorithms that are intuitive and practically relevant with only a small loss in performance compared to the state-of-the-art approaches.  }




\maketitle

%


%
\section{Introduction}
\label{sec:intro}

A common framework for combinatorial optimization that captures problems arising in wide-ranging applications is that of \emph{selecting a finite set of items} from a larger candidate pool and \emph{assigning these items to one or more groups}. Such problems form the core basis for the online content recommendation systems encountered in platforms pertaining to knowledge-sharing (e.g., Stack Overflow, Reddit), e-commerce~\citep{Li11}, and digital advertising as well as \emph{team selection} problems arising in gaming~\citep{GMH07} and traditional hiring. A crucial feature of these environments is the intrinsic uncertainty associated with the underlying items and consequently, sets of items. Given this uncertainty, the decision maker's objective in these domains is to maximize the expected group-value associated with the set of items and their assignment. 

As a concrete application, consider an online gaming platform where the items correspond to players; the platform may seek to assign (a subset of) players to teams in order to ensure competitive matches or to maximize the winning probability for a specific team. Other scenarios relating to team selection---e.g., a company hiring a set of candidates or a school identifying top students for a tournament---can also be modeled in an analogous fashion. Alternatively, these optimization problems arise in online communities such as Stack Overflow or Reddit. Here, the items represent topics or questions and the platform wishes to present a collection of relevant topics to an incoming user with the goal of maximizing that user's engagement measured via clicks or answers. Finally, in digital advertising, items may refer to ads displayed to a user in a marketing campaign and the value results from conversion events such as a click or product purchase. Naturally, all of these constitute stochastic environments due to the underlying uncertainty, e.g., the performance of any individual player is not deterministic in the case of a gaming platform, and there is considerable uncertainty regarding a user's propensity to click or respond to a topic on knowledge platforms.

There are several fundamental challenges in the above applications that necessitate innovative algorithmic approaches. First, the value derived from a set of items may not be linear in that of the individual items and may in fact, model a more subtle relationship. For example, agents or topics may complement or supplement each other; the efficiency of a team may grow with team size but exhibit diminishing returns as more members are added due to coordination inefficiencies. Second, the intrinsic uncertainty regarding the value of individual items may affect the group value in surprising ways due to the non-linearity of the objective. As we depict later, there are situations where a set of `high-risk high-reward' items may outperform a collection of stable-value items even when the latter type provides higher value in expectation. Finally, we also face issues relating to computational complexity since the number of items and groups can be very large in online platform scenarios and the underlying combinatorial optimization problems are usually NP-Hard.



Despite the above challenges, a litany of sophisticated algorithmic solutions have been developed for the problems mentioned previously. Given to the intricacies of the setting, these algorithms tend to be somewhat complex and questions remain on whether these methods are suitable for the scenarios outlined earlier owing to issues regarding scalability, interpretability, and the difficulties of function evaluation. On the other hand, it is common practice in many domains to select or assign items by employing algorithms that base their decisions on \emph{individual item scores}---these represent unique statistics associated with each item that serve as a proxy for the item's quality or the relevance to the task at hand. At a high level, these algorithms only use the scores computed for individual items---each item's score is independent of other items---to select items and as such, avoid the practical issues that plague traditional algorithmic paradigms.  

To expand on this thesis, consider a dynamic online portal such as Stack Overflow that hosts over eighteen million questions and wishes to recommend the most relevant subset to each incoming user. The platform may find it impractical to recompute the optimal recommended set of questions every time a new batch of questions is posted and thus, many traditional optimization methods are not scalable. At the same time, content recommendation services typically maintain relevance scores for each question and user-type pair that do not vary as new questions are posted and are utilized in practice to generate recommendation sets. In a similar vein, online gaming platforms estimate skill ratings (scores) for individual players based only on their past performance, which are in turn used as inputs for matchmaking. When it comes to team formation, these score based approaches may be preferable to standard algorithms that require \emph{oracle access} to the performance of every possible team. Indeed, evaluating the expected value of every subset of players even before the teams are formed seems prohibitively expensive.


Clearly, algorithms for selecting or assigning items based solely on individual item scores are appealing in many domains because of their conceptual and computational simplicity. However, a natural concern is that restricting the algorithmic landscape to these simple score based approaches may result in suboptimal solutions because they may be unable to account for complicated dependencies between individual item performance and the group output. Motivated by this tension, we study the following fundamental question:
\vspace*{2mm}
\begin{quote}
\emph{Can algorithms that assign items to groups based on individual item scores achieve near-optimal group performance and if so, under what conditions?}

\end{quote}
\vspace*{2mm}
We briefly touch upon our framework for stochastic combinatorial optimization. Let $N = \{1,2,\ldots,n\}$ be a ground set of items and let $2^N$ denote all possible subsets of $N$. Given a feasible set $\mathcal{F}\subseteq 2^N$ of items, a value function $f:2^N\times \reals^{n}\rightarrow \reals_+$, and a distribution $P$ of a random $n$-dimensional vector $\vec{X}=(X_1, X_2, \ldots, X_n)$, our goal is to select a set $S^*\in \mathcal{F}$ that is a solution to
\begin{equation}
\max_{S\in \mathcal{F}} u(S):=\E_{\vec{X}\sim P}[f(S,\vec{X})].
\label{equ:sop}
\end{equation}

In later sections, we generalize this formulation to consider problems where the goal is to select multiple subsets of $N$ and assign them to separate groups. The optimization problem (\ref{equ:sop}) is further refined as follows (see Section~\ref{sec:problem} for formal definitions):
\begin{enumerate}[label=(\alph*)]
\item We focus primarily on the canonical problem of \emph{maximizing a stochastic monotone submodular function subject to a cardinality constraint}. This is a special case of the optimization problem in (\ref{equ:sop}) where $\mathcal{F}$ is defined by the cardinality constraint $|S|= k$ for a given parameter $k$, and value function $f$ is such that the set function $u:2^N \rightarrow \reals_+$ is submodular. 

\item We restrict our attention to value functions $f$ where the output of $f(S,\vec{X})$ depends only on the elements of $\vec{X}$ that correspond to $S$, i.e., $(X_i)_{i \in S}$. Further, $X_i$ denotes the random performance of item $i \in N$ and is distributed independently of all other $X_j$ for $j \neq i$. Therefore, $P = P_1 \times P_2 \times \ldots \times P_n$ so that $X_i \sim P_i$.

\end{enumerate}

The framework outlined above captures a broad class of optimization problems arising in diverse domains. For example, submodular functions have featured in a variety of applications such as facility location~\citep{Ahmed2011}, viral influence maximization, job scheduling~\citep{cohen2019overcommitment}, content recommendation and team formation. In particular, submodularity allows us to model positive synergies among items and capture the natural notion of \emph{diminishing returns to scale} that is prevalent in so many situations---i.e., the marginal value derived by adding an item to a set cannot be greater than that obtained by adding it to any of its subsets. Moreover, in content recommendation as well as team selection, it is natural to expect that the performance of a group of elements $S$ would simply be a function (albeit a non-linear one) of the individual performances of the members in $S$---$(X_i)_{i \in S}$. This is represented by our assumptions on the value function $f$. 

The problem of maximizing a submodular function subject to a cardinality constraint is known to be NP-Hard and consequently, there is a rich literature on approximation algorithms for both the deterministic~\citep{krausesurvey} and stochastic variants~\citep{asadpourN16}.  In a seminal paper, \cite{NWF78} established that a natural greedy algorithm (sequentially selecting items that yield largest marginal value) guarantees a $1-1/e$ approximation of the optimum value, which is tight~\citep{F98}. Despite the popularity of greedy and other approaches, it is worth noting for our purposes that almost all of the algorithms in this literature are \emph{not robust to changes in the input}. That is, as the ground set $N$ grows, it is necessary to re-run the entire greedy algorithm to generate an approximately optimal subset. Furthermore, as mentioned earlier, these methods extensively utilize \emph{value oracle queries}---access to the objective function is through a black-box returning $u(S)$ for any given set $S$. 

\subsubsection*{Test Score Algorithms} 

We now formalize the notion of individual item scores, which we refer to henceforth, as \emph{test scores}. Informally, a test score is an item-specific parameter that quantifies the suitability of the item for the desired objective (i.e., $f$). To ensure scalability, it is crucial that an item's score depends only on the marginal distribution the item's individual performance and the problem specification. Formally, the test score $a_i \in [0,\infty)$ of an item $i \in N$ is defined as: 
\begin{equation}
a_i = h(f, \mathcal{F}, P_i), 
\label{equ:ts}
\end{equation}
where $h$ is a mapping from the item's marginal distribution ($P_i$), the objective value function $f$ and constraint set $\mathcal{F}$ to a single number. Naturally, there are innumerable ways to devise a meaningful test score mapping $h$. Obvious examples include: $(a)$ \emph{mean test scores} where $a_i = \E[X_i]$, and $(b)$ \emph{quantile test scores}, where $a_i$ is the $\theta$-quantile of distribution $P_i$ for some suitable $\theta$. However, we prove later that algorithms that base their decisions on these natural candidates do not always yield near-optimal solutions. 

The design question studied in this paper is to identify a suitable test score mapping rule $h$ such that algorithms that leverage these scores can obtain near-optimal guarantees for the problem defined in~\eqref{equ:sop}. Formally, a test score algorithm is a procedure that computes the test scores for each item in $N$ according to some mapping $h$ and uses only these scores to determine a feasible solution $S$ for~\eqref{equ:sop}, e.g., by selecting the $k$ items with the highest scores. Test score algorithms were first introduced by \cite{kleinberg2015team}, who developed algorithms for a team formation problem for a single specific function $f$. In this work, we propose a novel test score mechanism and utilize it to retrieve improved guarantees for a large class of naturally occuring functions. 

Test score algorithms are particularly salient in large-scale applications when compared to a more traditional optimization method such as greedy. First, as the ground set $N$ changes (e.g., posts are added or deleted), this does not alter the scores of items still present in the ground set since an item's test score depends only on its own performance distribution. This allows us to eliminate significant computational overhead in dynamic environments such as online platforms. Second, test score computations are trivially parallelizable---implemented via distributed computation---since each item's test score can be computed on a separate machine. Designing algorithms that are amenable to distributed implementation~\citep{balkanskiBS18} is a major concern nowadays and it is worth noting that standard greedy or linear programming approaches do not fulfill this criterion. Finally, test score algorithms allow us to make fewer and simpler oracle calls (function evaluations) as we highlight later. We now present a stylized formulation of a stochastic submodular optimization problem in an actual application in order to better illustrate the role of test scores.

%
%
%
%

\begin{example}{\textbf{(Content Recommendation on Stack Overflow or Reddit)}}
	\label{example:contentreco}
	The ground set $N$ comprises of topics created by users on the website. The platform is interested in selecting a set of $k$ topics from the ground set and present them to an arriving user in order to maximize satisfaction or engagement. For simplicity, the topics can be broadly classified into two categories---set $A$ consisting of useful but not very exciting topics and set $B$ which encapsulates topics that are polarizing or exciting\footnote{For instance, Reddit identifies certain posts as controversial based on the ratio of upvotes and downvotes}. Mathematically, we can capture this selection problem using our framework by taking $X_i$ to denote the utility that a user derives from topic $i \in N$ (alternatively $X_i$ could denote the probability of clicking or responding to a topic). For example, $X_i = a$ with probability one for $i \in A$ as these topics are stable, whereas $X_i = b/p$ with probability $p$ for each risky topic $i \in B$. The selection problem becomes particularly interesting when $b < a < b/p$. Due to cognitive limitations, one can assume that a user engages with at most $r \leq k$ topics from the assortment. Therefore, the objective function is defined as follows: $f(S,\vec{X}) = \sum_{j=1}^r X_{(j)}(S)$, where $X_{(j)}(S)$ refers to the $j$-th largest variable $X_i$ for $i \in S$. In the extreme case, $r=1$ and each user clicks on at most one topic. We refer to these as the \emph{top-$r$} and \emph{best-shot} functions respectively in Section~\ref{sec:problem}. \hfill $\blacksquare$
\end{example}

The tradeoff between `high-risk-high-reward' items and more stable items arises in a large class of selection problems in the presence of uncertainty. For example, in online gaming as in other team selection scenarios, a natural contention occurs between high performing players who exhibit a large variance (set $B$) and more consistent players (set $A$). In applications involving team formation, it is natural to use the \emph{CES (Constant Elasticity of Substitution)} utility function  as the objective, i.e., $f(S,\vec{X}) = (\sum_{i \in S}X^r_i)^{1/r}$, where the value of $r$ indicates the degree of substitutability of the task performed by the players~\citep{FSV16}. In this work, we design a natural test score based algorithm that allows us to obtain constant factor approximations for stochastic submodular optimization for all of the above objectives functions.

\subsection{Main Contributions}

The primary conceptual contribution of this study is the introduction of a framework for analysis of test score based algorithms for stochastic combinatorial optimization problems involving selection and assignment. We believe that this paradigm helps bridge the gap between theory and practice, particularly in large-scale applications where quality or relevance scores are prominently used for optimization. For these cases, the mechanisms developed in this work provides a rigorous framework for computing and utilizing these scores. 

Our main technical contribution is the design of a test score mapping which gives us good approximation algorithms for two NP-Hard problems, namely: (a) maximizing a stochastic monotone submodular function subject to a cardinality constraint, and (b) maximizing a stochastic submodular welfare function, defined as a sum of stochastic monotone submodular functions subject to individual cardinality constraints. The welfare maximization problem is a strict generalization of the former and is of interest in online platforms, where items are commonly assigned to multiple groups, e.g., selection of multiple disjoint teams for an online gaming tournament. 

We now highlight our results for the first problem. We identify a special type of test scores that we refer to as \emph{replication test scores} and show that under a sufficient condition on the value function (extended diminishing returns), we achieve a \emph{constant factor approximation for the problem of maximizing a stochastic submodular function subject to a cardinality constraint}. At a high level, replication test scores can be interpreted as a quantity that measures both an item's individual performance as well its marginal contribution to larger team of equally skilled items---see Section~\ref{sec:SFM} for a formal treatment. Additionally, we also show the following:
\begin{itemize}
	\item We provide an intuitive interpretation of the extended diminishing returns property and prove that it is satisfied by a number of naturally occuring value functions including but not limited to the functions mentioned in our examples such as best-shot, top-$r$, and CES.
	
	\item  We show that replication scores enjoy a special role in the family of all feasible test scores: in particular, for any given value function, if there exist any test scores that guarantee a constant factor approximation for the submodular maximization problem, then it is possible to obtain a constant factor approximation using replication test scores. This has an important implication that in order to find good approximation factors, it suffices to consider replication test scores. 
	
	
	\item We highlight cases where natural test score measures such as mean and quantile test scores do not yield a constant factor approximation. We provide a tight characterization of their efficiency for the CES function---specifically, mean test scores provide only a $1/k^{1-1/r}$-approximation to the optimum and quantile scores do not guarantee a constant-factor approximation when $r < \Theta(\log(k))$. Recall that $r$ denotes the degree of substitutability among items.
	

\end{itemize}





Finally, for the more general problem of stochastic submodular welfare maximization subject to cardinality constraints, with the value functions satisfying the extended diminishing returns condition, we establish that replication test scores guarantee a $\Omega(\frac{1}{\log(k)})$-approximation to the optimum value, where $k$ is the maximum cardinality constraint. This approximation is achieved via a slightly more intricate algorithm that greedily assigns items to groups based on their replication test scores. 

Our results are established by a novel framework that can be seen as approximating (sketching) set functions using test scores. In general, a sketch of a set function is defined by two simpler functions that lower and upper bound the original set function within given approximation factors. In our context, we present a novel construction of a sketch that only relies on replication test scores to approximate a submodular function \emph{everywhere}. By leveraging this sketch, we show that selecting the $k$ items with the highest test scores is only a constant factor smaller than the optimal set. These results may be of independent interest.  

\subsection{Related Work} 
The problem of maximizing a stochastic submodular function subject to a cardinality constraint by using test scores was first posed by \cite{kleinberg2015team} who developed constant factor approximation algorithms but only for a specific value function, namely the top-$r$ function. They introduced the term `test scores' in the context of designing algorithms for team hiring to indicate that the relevant score for each candidate can often be measured by means of an actual test. Their work also provides some impossibility results, namely that test score algorithms cannot yield desirable guarantees for certain submodular functions. Our work differs in several respects. First, we show that test scores can guarantee a constant factor approximation for a broad class of stochastic monotone submodular functions, which includes different instances of value functions used in practice. Second, we extend this theoretical framework to the more general problem of stochastic submodular welfare maximization, and obtain novel approximation results by using test scores. Third, we develop a unifying and systematic framework based on approximating set functions by simpler test score based sketches. 

As we touched upon earlier, submodular functions are found in a plethora of settings and there is a rich literature on developing approximation algorithms for different variants of the cardinality-constrained and welfare maximization problems~\citep{lehmann2006combinatorial,vondrak08}. Commonly used algorithmic paradigms for these problems include greedy, local search, and linear programming (with rounding). Due to their reliance on these sophisticated techniques, most if not all of these algorithms are (a) not scalable in dynamic environments as the algorithm has to be fully re-executed every time the ground set changes, and (b) hard to implement in a parallel computing model. More importantly, these policies are inextricably tied to the value oracle model and hence, tend to query the oracle a large number of times; often these queries are aimed at evaluating the function value for arbitrary subsets of the ground set. As we illustrate in Section~\ref{sec:oracles}, oracle queries can be expensive in certain cases. On the other hand, the test score algorithm proposed in this work makes use of much fewer oracle queries. Within the realm of submodular maximization, there are three distinct strands of literature that seek to tackle each of the three issues mentioned above. 
\begin{itemize}
\item \emph{Dynamic Environments}: A growing body of work has sought to develop \emph{online algorithms} for submodular and welfare maximization problems in settings where the elements of the ground set arrive sequentially~\citep{feldmanZ18,korula2018online} In contrast to this work, the decisions made by online algorithms are irrevocable, where test score algorithms are only aimed at reducing the computational burden when the ground set changes. 

\item \emph{Distributed Implementation}: Following the rise of big data applications and map-reduce models, there has been a renewed focus on developing algorithms for submodular optimization that are suitable for parallel computing. The state-of-the-art (distributed) algorithms for submodular maximization are $O(\log(n))$-adaptive---they run for $O(\log(n))$ sequential rounds with parallel computations in each round~\citep{balkanskiBS18, fahrbachMZ19}. Since each test score can be computed independently, our results can be interpreted as identifying a well-motivated special class of submodular functions which admit $1$-adaptive algorithms. 

\item \emph{Oracle Queries}: The design of test score algorithms is well-aligned with the body of work on maximizing submodular set functions using a small number of value oracle queries~\citep{badanidiyuruMKK14,BS18,fahrbachMZ19}. In fact, our replication test score based algorithms only query the function value for subsets comprising of similar or identical items. 
\end{itemize}
Although there is a promising literature pertaining to each of these three challenges, our test score based techniques represent the first attempt at addressing all of them. While many of the above papers propose algorithms for deterministic environments, recently, there has been considerable focus on maximizing submodular functions in a stochastic setting~\cite[e.g.,][]{KS17,STK16,asadpourN16,gotovosHK15,kempeKT15,Asadpour2008}. However, the methods presented in these works do not address any of the concerns mentioned earlier and to a large extent,  focus explicitly on settings where it is feasible to \emph{adaptively probe} items of the ground set to uncover the realization of their random variable ($X_i$). 

More generally, a powerful paradigm for solving stochastic optimization problems as defined in~\eqref{equ:sop} is the technique of \emph{Sample Average Approximation} (SAA)~\citep{KSH02,SN05,swamy2012sampling}. These methods are typically employed when the following conditions are applicable, see e.g. \cite{KSH02}: (a) the function $u(S)$ cannot be written in a closed form, (b) the value of the function $f(S,\vec{x})$ can be evaluated for every given set $S$ and vector $\vec{x}$, and (c) the set $\mathcal{F}$ of feasible solutions is large. The fundamental principle underlying this technique is to generate samples $(\vec{x}^{(1)}, \ldots, \vec{x}^{(T)})$ independently from the distribution $P$ and use these to compute the set $S^*$ that is the optimal solution to $\arg\max_{S \in \mathcal{F}} \frac{1}{T}\sum_{i=1}^T f(S,\vec{x}^{(i)})$. 

In addition to the same drawbacks regarding scalability mentioned above, there are other situations where it may be advantageous to use a test score algorithm over SAA methods: (a) when the function $f$ is accessed via a value oracle, a large number of queries may be required to optimize the sample-approximate objective, and (b) even if oracle access is not a concern and the underlying function is rather simple (e.g., best-shot function from Example~\ref{example:contentreco}), computing the optimal set $S^*$ may be NP-Hard (see Appendix~\ref{app:saa}). 
Finally, by means of numerical simulations in Section~\ref{sec:disc1}, we highlight well-motivated scenarios where SAA methods may result in a higher error probability compared to test score algorithms under the same number of samples drawn.



The techniques in our work are inspired by the theory on set function sketching~\citep{goemansHIM09,balcanH11,CD17}, and their application to optimization problems~\citep{iyerB13}. While the $\Omega(1/\sqrt{n})$ sketch of~\cite{goemansHIM09} for general submodular functions does apply to our setting, we are able to provide tighter bounds (the logarithmic bound of Lemma~\ref{thm:SMBsketch}) for a special class of well-motivated submodular functions that cannot be captured by existing frameworks such as curvature~\citep{sviridenkoVW15}. Our approach is also similar in spirit to~\citet{iyerB13}, where upper and lower bounds in terms of so-called \textit{surrogate functions} were used for submodular optimization; the novelty in the present work stems from our usage of test scores for function approximation, which are conceptually similar to \emph{juntas}~\citep{feldmanV14}. We believe that the intuitive and natural interpretation of test score-based algorithms make them an appealing candidate for other problems as well. 

	
\subsection{Organization of the Paper} 

The paper is structured as follows. Section~\ref{sec:problem} provides a formal definition of optimization problems studied in this paper and introduces examples of value functions. Section~\ref{sec:SFM} contains our main result for the problem of maximizing a stochastic montotone submodular function subject to a cardinality constraint. Section~\ref{sec:subm-welf-maxim} contains our main result for the problem of maximizing a stochastic monotone submodular welfare function subject to cardinality constraints. Section~\ref{sec:disc} presents a numerical evaluation of a test score algorithm for a simple illustrative example, a tight characterization of approximation guarantees achieved by mean and quantile test scores for the CES value function, and some discussion points. Finally, we conclude in Section~\ref{sec:conc}. All the proofs of theorems and additional discussions are provided in Appendix.	
\section{Model and Problem Formulation} 
\label{sec:problem}

In this section, we introduce basic definitions of submodular functions, more formal definitions of the optimization problems that we study, and examples of various value functions. 

\subsection{Preliminaries: Submodular Functions}
\label{sec:prelim}
Given a ground set $N = \{1,2,\ldots, n\}$ of items or elements with $2^N$ being the set of all possible subsets of $N$, a set function $u:2^N \rightarrow \reals_+$ is \emph{submodular} if $u(S \cup T) + u(S \cap T) \leq u(S) + u(T)$, for all $S,T \in 2^N$. This condition is equivalent to saying that $u$ satisfies the intuitive \emph{diminishing returns property}: $u(T\cup \{i\}) - u(T) \leq u(S\cup \{i\}) - u(S)$ for all $i\in N$ and $S,T\in 2^N$ such that $S\subseteq T$.  Furthermore, we say that $u$ is \emph{monotone} if $u(S)\leq u(T)$ for all $S,T\in 2^N$ such that $S\subseteq T$.


Next, we adapt the definition of a stochastic submodular function, e.g. used in~\citep{asadpourN16}, as the expected value of a submodular value function. Let $g:\reals^n\rightarrow \reals_+$ be a \emph{value function} that maps $n$-dimensional vectors to non-negative reals---$g$ is said to be a \emph{submodular value function} if for any two vectors $\vec{x}, \vec{y}$ belonging to its domain:
\begin{equation}
g(\vec{x} \vee \vec{y}) + g(\vec{x} \wedge \vec{y}) \leq g(\vec{x}) + g(\vec{y}).
\label{eqn_subvalue}
\end{equation}

In the above definition, $\vec{x} \vee \vec{y}$ denotes the component-wise maximum and $\vec{x} \wedge \vec{y}$ the component-wise minimum. Note that when the domain of $g$ is the set of Boolean vectors (all elements taking either value $0$ or $1$), then~\eqref{eqn_subvalue} reduces to the definition of a submodular set function. Hence, submodular value functions are a strict generalization of submodular set functions. Finally, we say that the value function $g$ is monotone if for any two vectors $x$ and $y$ satisfying $\vec{y} \geq \vec{x}$ ($\vec{y}$ dominates $\vec{x}$ component-wise), we have $g(\vec{y}) \geq g(\vec{x})$. 

Consider the ground set $N$ and for every $S\in 2^N$, we define $\vec{x}\mapsto \perfmap_S(\vec{x})$ to be a mapping such that $M_S(\vec{x})_i = x_i$ if $i\in S$ and $M_S(\vec{x}) = \phi$, otherwise. Here, $\phi$ is a minimal element which does not change the function value by adding an item of individual value $\phi$. For example, for the mapping $g(\vec{x}) = \max\{x_1, x_2, \ldots, x_n\}$, we may define $\phi=0$. When it is clear from the context, we sometimes abuse notation by writing $g(\vec{x})$ for a vector of dimension $d < n$, instead of $g(\vec{x}, \vec{z})$ where $\vec{z}$ is a vector $\vec{x}$ of dimension $n-d$ that has all elements equal to $\phi$. Now, we are ready to define stochastic submodular functions. Suppose that each item $i \in N$ is associated with a non-negative random variable $X_i$ that is drawn independently from distribution $P_i$.  We assume that each $P_i(x)$ is a cumulative distribution function, i.e. $P_i(x)=\Pr[X_i \leq x]$. Given a monotone submodular value function $g$, a set function $u:2^N \rightarrow \reals_+$ is said to be a stochastic monotone submodular function if for all $S \in 2^N$:
\begin{equation}
u(S) = \E[g(\perfmap_S(X_1, X_2, \ldots,X_n))].
\label{eqn_stochasticsub}
\end{equation}

For example, if $g$ is the max or best-shot function, then $u(S) = \E[\max_{i \in S}\{X_i\}]$. The following result, which we borrow from Lemma~3 in~\cite{asadpourN16}, provides sufficient reasoning on why it is accurate to interpret $u$ to be submodular.




\begin{lemma} \label{lem_stochasticsubmodularity} Suppose that $g$ is a monotone submodular value function. Then, a set function $u$ that is defined as in~\eqref{eqn_stochasticsub} is a monotone submodular set function.
\end{lemma}

\subsection{Problem Definitions} 
\label{sec:sfm}
In this work, we study the design of test score algorithms for two combinatorial optimization problems, namely: (a) maximizing a stochastic monotone submodular function subject to a cardinality constraint, and (b) maximizing a stochastic monotone submodular welfare function defined as the sum of stochastic monotone submodular functions subject to cardinality constraints. We begin with the first problem. Recall the optimization problem presented in~\eqref{equ:sop}, and suppose that $\mathcal{F} = \{S \subseteq N \mid |S| = k\}$ for a given cardinality constraint $0 < k \leq n$ and let $\vec{X} = (X_1, \ldots, X_n)$ be a vector of random, independently and not necessarily identically distributed item performances such that for each $i \in N$, $X_i \sim P_i$. By recasting problem~\eqref{equ:sop} in terms of the notation developed in Section~\ref{sec:prelim}, we can now define the problem of maximizing a stochastic monotone submodular function subject to a cardinality constraint $k$ as follows\footnote{We used the formulation $u(S) = \E[f(S,\vec{X})]$ in the introduction to maintain consistency with the literature on stochastic optimization, e.g.,~\citep{KSH02}. For the rest of this paper, we will exclusively write $u(S) = \E[g(\perfmap_S(\vec{X}))]$ for convenience and to delineate the interplay between the set $S$ and the submodular value function $g$. }
\begin{equation}
\argmax_{S \in \mathcal{F}} u(S) := \E[f(S,\vec{X})] := \E[g(\perfmap_S(\vec{X}))], 
\label{equ:ug}
\end{equation}

where $g$ is a monotone submodular value function. Additionally, we assume the function $g$ to be \emph{symmetric}, meaning that its value is invariant to permutations of its input arguments, i.e. for every $\vec{x}\in \reals^n$, $g(\vec{x}) = g(\pi(\vec{x}))$ for any permutation $\pi(\vec{x})$ of the elements $\{x_1, x_2, \ldots, x_n\}$. This is naturally motivated by scenarios where the group value of a set of items depends on the individual performance values than the identity of the members who generate these values. For example, in the case of non-hierarchical team selection, it is reasonable to argue that two mutually exclusive teams $S, T$ whose members yield identical performances on a given day also end up providing the same group value. Similarity, in content recommendation, the probability that user clicks on at least one topic can be viewed as a function of the user's propensity to click on each individual topic. Finally, by seeking to optimize the expected function value in~\eqref{equ:ug}, we implicitly model a risk-neutral decision maker as is typically the case in online platforms. 
 
The stochastic submodular maximization problem specified in~\eqref{equ:ug} is NP-Hard even when the value function $g$ is symmetric (in fact,~\cite{goelGM06} show this is true for $g(\vec{x}) = \min\{x_1, \ldots, x_n\}$), and hence, we focus on finding approximation algorithms. Formally, given $\alpha \leq 1$, an algorithm is said to provide an $\alpha$-approximation guarantee for~\eqref{equ:ug} if for any given instance of the problem with optimum solution set $\mathrm{OPT}$, the solution $S$ returned by the algorithm satisfies $u(S) \geq \alpha u(\mathrm{OPT})$. Although a variety of approximation algorithms have been proposed for the submodular maximization problem, in this work, we focus on a special class of methods we refer to as \emph{test score algorithms}. Specifically, these are algorithms that take as input a vector of non-negative test scores $a_1, a_2, \ldots, a_n$, and use only these scores to determine a feasible solution $S$ for the problem~\eqref{equ:ug}. As defined in~\eqref{equ:ts}, the value of each test score $a_i$ can depend only on $g$, $k$, and $P_i$. Furthermore, we are particularly interested in proposing test score algorithms that simply select the $k$ items with the highest test scores in $(a_1, a_2, \ldots, a_n)$; such an approach is naturally appealing due to its intuitive interpretation. Clearly, the main challenge in this case is to design a suitable test score mapping rule that enables such a trivial algorithm to yield desirable guarantees.



\subsubsection*{Stochastic Submodular Welfare Maximization} Maximizing a stochastic submodular welfare function is a strict generalization of the problem of maximizing a stochastic monotone submodular function subject to a cardinality constraint as defined in~\eqref{equ:ug}. Here, we are given a ground set $N = \{1,\ldots, n\}$, and a collection of stochastic monotone submodular set functions $u_j : 2^N \rightarrow \reals_+$ with corresponding submodular value functions $g_j:\reals^n\rightarrow \reals_+$ for $j \in M := \{1,2,\ldots, m\}$. The goal is to find disjoint subsets $S_1, S_2, \ldots, S_m$ of the ground set of given cardinalities $|S_1| = k_1$, $|S_2| = k_2$, $\ldots$, $|S_m|=k_m$ that maximize the welfare function defined as
\begin{equation}
u(S_1, S_2, \ldots, S_m) = \sum_{j=1}^m u_j(S_j). 
\label{equ:smp}
\end{equation}
We refer to $M$ as the set of partitions. Similarly as for the previous problem, we consider symmetric, monotone, submodular value functions $g_j$ for each partition $j \in M$ so that the stochastic submodular set functions can be represented as follows: 
$$
u_j(S) = \E[g_j(M_S(X_{1,j}, X_{2,j},\ldots, X_{n,j}))] \quad \text{for all $j \in M$}.
$$
In the above expression, $X_{i,j}$ denotes the individual performance of item $i \in N$ with respect to partition $j \in M$. Each $X_{i,j}$ is drawn independently from a marginal distribution $P_{i,j}$ that is the cumulative distribution function $P_{i,j}(x) = \Pr[X_{i,j}\leq x]$. Our formulation allows for considerable heterogeneity as items can have different realizations of their individual performances for different partitions. Submodular welfare maximization problems arise naturally in domains such as team formation where decision makers are faced with the dual problem of selecting agents and assigning them to projects or teams. For example, this could model an online gaming platform seeking to choose a collection of teams to participate in a tournament or an organization partitioning its employees to focus on asymmetric tasks. In these situations, the objective function~\eqref{equ:smp} captures the aggregate value generated by all of the teams. 

Once again, we are interested in designing test score algorithms for stochastic submodular welfare maximization. Due to the generality of the problem, we define test score based approaches in a broad sense here and defer the specifics to Section~\ref{sec:subm-welf-maxim}. More formally, a test score algorithm for problem~\eqref{equ:smp} is a procedure whose input only comprises of vectors of test scores $(\vec{a}_{i,j})_{i \in N, j \in M}$, where the elements of each test score vector $\vec{a}_{i,j}$ are a function of $g_j, k_j$, and $P_{i,j}$. Note that in this general formulation, each item $i \in N$ and partition $j \in M$ is associated with multiple test scores $\vec{a}_{i,j}$.

\subsection{Examples of Value Functions}
\label{sec:val}

Many value functions used in literature to model production and other systems satisfy the conditions of being symmetric, monotone non-decreasing submodular value functions. In this section, we introduce and discuss several well known examples. 

A common value function is defined to be an increasing function of the sum of individual values: $g(\vec{x}) = \bar{g}\left(\sum_{i=1}^n x_i\right)$, where $\bar{g}$ is a non-negative increasing function. In particular, this value function allows to model production systems that exhibit a diminishing returns property when $\bar{g}$ is concave. This value function appears frequently in optimization problems when modeling risk aversion and decreasing marginal preferences,
for instance, in risk-averse capital budgeting under uncertainty, competitive facility
location, and combinatorial auctions \citep{Ahmed2011}. A popular example of such a function is the threshold or budget-additive function, i.e., $g(\vec{x}) = \min\{\sum_{i=1}^n x_i, B\}$ for some $B > 0$, which arises in a number of applications. 

Another example is the \emph{best-shot value function} defined as the maximum individual value $g(\vec{x}) = \max\{x_1,x_2,\ldots,x_n\}$. This value function allows to model scenarios when one only derives values from the best individual option. For example, this arises in online crowdsourcing systems in which solutions to a problem are solicited by an open call to the online community, several candidate solutions are received, but eventually only a best submitted solution is used. 

A natural generalization of the best-shot value function is a \emph{top}-$r$ \emph{value function} defined as the sum of $r$ highest individual values, for a given parameter $r\geq 1$, i.e. $g(\vec{x}) = x_{(1)} + x_{(2)} + \cdots + x_{(r)}$, where $x_{(i)}$ is the $i$-th largest element of input vector $\vec{x}$. This value function boils down to the best-shot value function for $r = 1$. This value function is of interest in a variety of applications such as information retrieval and recommender systems, where the goal is to identify a set of most relevant items. This value function was used in \cite{kleinberg2015team} to evaluate efficiency of test-score based algorithms for maximizing a stochastic monotone submodular function subject to a cardinality constraint. 

A well known value function is the \emph{constant elasticity of substitution (CES) value function}, which is defined by $g(\vec{x}) = \left( \sum_{i=1}^n x_i^r \right)^{1/r}$, for a positive value parameter $r$. This value function has been in common use to model production systems in economics and other areas~\citep{FSV16,DS77,A69,S56}. The family of CES value functions accommodates different types of production by suitable choice of parameter $r$, including the linear production for $r = 1$ and the best-shot production in the limit as the value of parameter $r$ goes to infinity. The CES value function is a submodular value function for values of parameter $r\geq 1$. For $r \neq 1$, the term $1/(1-r)$ is referred to as the elasticity of substitution---it is the elasticity of two input values to a production with respect to the ratio of their marginal products. 

Finally, we make note of the \emph{success probability value function}, defined by $g(\vec{x}) = 1-\prod_{i=1}^n(1-p(x_i))$, where $p:\reals \rightarrow [0,1]$ is an increasing function that satisfies $p(0) = 0$. This value function is often used as a model of tasks for which input solutions are independent and either good or bad (success or failure), and it suffices to have at least one good solution for the task to be successfully solved, e.g., see~\cite{kleinbergO11}. 

\subsection{Computation, Implementation, and the Role of Value Oracles}
\label{sec:oracles}


We conclude this section with a discussion of some practical issues surrounding test scores algorithms and function evaluation. Given that submodular set functions have representations that are exponential in size $(2^n)$, a typical modeling concession is to assume access to a value oracle for function evaluation. Informally, a value oracle is a black-box that when queried with any set $S \in 2^N$, returns the function value $u(S)$ in constant time. Although value oracles are a theoretically convenient abstraction, function evaluation can be rather expensive in applications pertaining to online platforms. This problem is further compounded in the case of stochastic submodular functions when the underlying item performance distributions $(P_1, P_2, \ldots, P_n)$ are unknown. Naturally, one would expect a non-zero query-cost to be associated with evaluating $g(\vec{x})$ even for a single realization $\vec{x}$ of the random vector $\vec{X}$. Under these circumstances, there is a critical need for algorithms that achieve desirable guarantees using significantly fewer queries and to eschew traditional approaches (e.g., greedy) that require polynomially many oracle calls. 

To illustrate these challenges, consider the content recommendation application from Example~\ref{example:contentreco} and suppose that both the distributions $(P_i)_{i \in N}$ and the value function $g$ are unknown. In order to (approximately) compute $u(S)$ for any $S \subseteq N$, it is necessarily to present the set $S$ of topics repeatedly to a large number of users and average their response (e.g., upvotes or click behavior). Clearly, a protracted experimentation phase brought about by too many oracle queries could lead to customer dissatisfaction or even a loss in revenue. Alternatively, in team hiring or online gaming, evaluating the function value for arbitrary subsets $S \subseteq N$ may be prohibitively expensive as it may not be possible to observe group performance before the team is even formed. The replication test score algorithm proposed in Section~\ref{sec:SFM} addresses these issues by not only making use of fewer oracle calls but also allowing for easier implementation since each evaluation of the function $g$ only requires samples from a single item's (or agent's) performance distribution $P_i$.

A secondary issue concerns the noise in the function evaluation or test score computation brought about by sampling the distributions $(P_i)_{i \in N}$. It may not be possible to precisely compute test scores $a_i$ that represent the expected value of some function under distribution $P_i$---e.g., mean test scores where $a_i = \E_{X_i\sim P_i}[X_i]$ or replication test scores in~\eqref{eqn:replscores}. In applications, test scores are defined as sample estimators with values determined by the observed data, i.e., utilize a sample mean instead of the population mean. In our analysis, we ignore the issue of estimation noise and assume oracle access that facilitates the precise computation of test scores that denote some expectation taken over $(P_i)_{i \in N}$. This assumption is justified provided that the estimators are unbiased and the test scores are estimated using a sufficient number of samples. We leave accounting for statistical estimation noise as an item for future research. 

Finally, it is worth highlighting that the benefits of test score algorithms do not come without a price. Using a test score based approach severely limits what an algorithm can do, which in turn may affect the achievable quality of approximation. For instance, the aforementioned greedy algorithm is able to leverage its unrestricted access to a value oracle and achieve a $1-1/e$-approximation for~\eqref{equ:ug} by carefully querying the function value for many different subsets $S \in 2^N$. Test score algorithms, however, do not have this luxury---instead, they rely indirectly on approximating answers to value oracle queries using only limited information, namely parameters associated with individual items $i \in N$ evaluated separately on the function $g$. 	
%

\section{Submodular Function Maximization}
\label{sec:SFM}

In this section we present our main result on the existence of test scores that guarantee a constant-factor approximation for maximizing a stochastic monotone submodular function subject to a cardinality constraint, for symmetric submodular value functions that satisfy an extended diminishing returns condition. We will show that this is achieved by special type of test scores.


We begin by introducing some basic terminology required for our sufficient condition. Given a value function $g: \reals_+^n \rightarrow \reals_+$ and $v \geq 0$, we say that $v$ has a non-empty \emph{preimage} under $g$ if there exists at least one $\vec{z} \in \reals_+^n$ such that $g(\vec{z}) = v$.

\begin{definition}[Extended Diminishing Returns] A symmetric submodular value function $g:\reals_+^n\rightarrow \reals_+$ is said to satisfy the \emph{extended diminishing returns} property if for every $v \geq 0$ that has a non-empty preimage under $g$, there exists $\vec{z} \in \reals_+^{n-1}$ such that:
		
\begin{itemize}
\item[(a)] $g(\vec{z}) = v$,
and
\item[(b)] for all $\vec{y}\in \reals^{n-1}_+$ such that $g(\vec{y})\leq v$, we have that $g(\vec{y},x) - g(\vec{y}) \geq g(\vec{z},x) - g(\vec{z}) \hbox{ for all } x\in \reals_+$.
\end{itemize}
\label{def:esa}
\end{definition}


Informally, the condition states that given that a value $v$ such that the function evaluates to this number at one or more points in its domain, then for at least one such point, say $\vec{z}$, the marginal benefit of adding an element of value $x$ to $\vec{z}$ cannot be larger than the marginal benefit of adding the same element to another vector $\vec{y}$ whose performance is smaller than $v$. The extended submodularity condition holds for a wide range of functions.  For example, the condition is satisfied by all value functions defined and discussed in Section~\ref{sec:val}, which is proved in Appendix~\ref{sec:check}.

We refer to this property as extended diminishing returns as it is consistent with the spirit of `decreasing marginal returns' as the function value grows. Indeed, as in the case of traditional submodular functions, adding a new element ($x$) provides greater marginal benefit to a vector yielding a smaller performance ($\vec{y}$) than to one providing a larger value $(\vec{z})$. In other words, we have diminishing marginal returns as the value provided by a vector $\vec{z}$ grows. Consider for example, a team application: a new member with some potential would be expected to make a less significant contribution to a high performing team than a low performing one. Similarly, in content recommendation, the added benefit provided by a new topic would be felt more strongly by a user who derives limited value from the original assortment than one who was highly satisfied to begin with. The underlying mechanism in both these examples is that a new member or topic would have a greater overlap in skills or content with a high performing group of items. 


A subtle point is worth mentioning here. For any given $v$, if there exist multiple points in the domain at which the function $g$ evaluates to $v$, then the extended diminishing returns property only guarantees the existence of a single vector $\vec{z}$ for which $g(\vec{z},x) - g(\vec{z}) \geq g(\vec{y},x) - g(\vec{y})$ holds for all $\vec{y}, x$ such that $g(\vec{y}) \leq v$. Simply put, there may be other vectors which also evaluate to $v$ which do not satisfy the above inequality.\footnote{Suppose that $g$ is the top$-r$ function defined in Section~\ref{sec:val} for $r=2$, $\vec{x} = (1,1)$, and $v = 4$. Consider vectors $\vec{y}_1 = (2,2)$ and $\vec{y}_2 = (4)$ such that $g(\vec{y}_1) = g(\vec{y}_2) = v = 4$. It is not hard to deduce that for any  $0 < z \leq 1$, $g(\vec{x},z) - g(\vec{x}) = 0 = g(\vec{y}_1,z) - g(\vec{y}_1) < g(\vec{y}_2,z) - g(\vec{y}_2) = z$. That is $\vec{y}_1$ satisfies the conditions in Definition~\ref{def:esa} but $\vec{y}_2$ does not.}
 We remark that this is actually a weaker requirement than imposing that all such vectors satisfy condition (b) in Definition~\ref{def:esa}---this allows our results to be applicable for a broader set of functions. That being said, most of the value functions that we specify in Section~\ref{sec:val} except for top-$r$ ($r > 1$) satisfy a stronger version of extended diminishing returns where the condition $g(\vec{z}, x) - g(\vec{z}) \geq g(\vec{y},x) - g(\vec{y})$ holds for every two points $\vec{z},\vec{y}\in \reals_+^{n-1}$ such that $g(\vec{y})\leq g(\vec{z})$.


We next introduce the special type of test scores, we refer to as replication test scores. 

\begin{definition}[replication test scores] Given a symmetric submodular value function $g$ and cardinality parameter $k$, for every item $i\in N$, the replication test score $a_i$ is defined by
\begin{equation}
\label{eqn:replscores}
a_i = \E[g(X_i^{(1)}, X_i^{(2)},\ldots, X_i^{(k)},\phi,\ldots,\phi)]
\end{equation}
where $X_i^{(1)}, X_i^{(2)},\ldots, X_{i}^{(k)}$ are independent and identically distributed random variables with distribution $F_i$. 
\label{def:rep}
\end{definition}

The replication test score of an item can be interpreted as the expected performance of a \emph{virtual} group of items that consists of $k$ independent replicas of this item, hence the name replication scores. Note that a replication test score is defined for a given function $g$ and cardinality parameter $k$; we omit to indicate this in the notation $a_i$ for simplicity. 

In contrast to mean or quantile test scores that simply provide some measure of an item's performance, replication test scores capture both the item's individual merit as well as its contribution to a larger group. To understand this distinction, consider Example~\ref{example:contentreco} where $g(\vec{x}) = \max\{x_1, x_2, \ldots, x_n\}$ and $p=1/k$. Clearly, the mean performance of stable type $A$ items ($a$) is larger than the mean performance of polarizing topics of type $B$ ($b$). However, the replication score of a type $B$ item is $(1 - (1-p)^k)\frac{b}{p} \geq (1-\frac{1}{e}) b k$ which for large enough $k$ can be  larger than the replication score of a type $A$ item which still remains $a$. The larger replication score of type $B$ topics captures the intuition that risky topics can often provide significant marginal benefits to an existing assortment. Finally, in the case of content recommendation, one can employ a natural mechanism to estimate the replication scores even when the objective function $g$ and distributions $(P_i)_{i \in N}$ are unknown. Namely, in order to compute the replication score for a topic of type $A$ (or $B$), it suffices to present $k$ items of this type to a large number of incoming users and compute the average response. 

We now present the main result of this section. 

\begin{theorem} Suppose that the utility set function is the expected value of a symmetric, monotone submodular value function that satisfies the extended diminishing returns property. Then, the greedy selection of items in decreasing order of replication test scores yields the utility value that is at least $(1-1/e)/(5-1/e)$ times the optimal value.   
\label{thm:subtest-sum}
\end{theorem}

In the remainder of this section, we prove Theorem~\ref{thm:subtest-sum}. Along the way, we derive several results that connect the underlying discrete optimization problem with approximating set functions, which may be of independent interest.

\label{sec:sbfp}

The key mathematical concept that we use is a \emph{sketch} of a set function, which is an approximation of a potentially complicated set function using simple polynomial-time computable lower and upper bound set functions, we refer to as a \emph{minorant} and a \emph{majorant} sketch function, respectively.  
	
\begin{definition}[Sketch] A pair of set functions $(\underline{v},\bar{v})$ is said to be a $(\clow,\chigh)$-sketch of a set function $u:2^N\rightarrow \reals_+$, if the following condition holds:
\begin{equation}
\clow \underline{v}(S) \leq u(S) \leq \chigh \bar{v}(S), \hbox{ for all } S\subseteq N.
\end{equation}
In particular, if $(v,v)$ is a $(\clow,\chigh)$-sketch, we refer to $v$ as a \emph{strong sketch} function.\footnote{Our definition of a strong sketch is closely related to the following definition of a sketch used in literature (e.g., see \cite{CD17}): a set function $\tilde{v}$ is said to be a $\alpha$-sketch of $u$ if $\tilde{v}(S)\leq u(S)\leq \alpha \tilde{v}(S)$ for all $S\subseteq N$. Indeed, if $v$ is a $(p,q)$-strong sketch of $u$, then $\tilde{v}(S):=pv(S)$ is a $q/p$-sketch of $u$.} 
\label{def:weaksketch}
\end{definition}
	
Although the above definition is quite general, and subsumes many trivial sketches (for e.g, $\underline{v} = 0, \bar{v} = \infty$), practically useful sketches would satisfy a few fundamental properties such as (a) when given a set function whose description may be exponential in $n$, $\underline{v}$ and $\bar{v}$ must be polynomially expressible, and (b) $\underline{v}$ and $\bar{v}$ must be sufficiently close to each other at points of interest for the sketch to be meaningful. Our first result provides sufficient conditions on the sketch functions to obtain an approximation algorithm for maximizing a monotone submodular set function subject to a cardinality constraint.
	
\begin{lemma} Suppose that (a) $\underline{v}$ and $\bar{v}$ are minorant and majorant set functions that are a $(\clow,\chigh)$-sketch of a submodular set function $u:2^N\rightarrow \reals_+$ and (b) there exists $S^* \subseteq \argmax_{S:|S|=k}\underline{v}(S)$ that satisfies $\bar{v}(S) \leq \underline{v}(S^*)$ for every $S\subseteq N$ that has cardinality $k$ and is completely disjoint from $S^*$, i.e. $S\cap S^* = \emptyset$. Then, the following relation holds: 
$$
u(S^*) \geq \frac{\clow}{\chigh+\clow}u(\OPT), 
$$
where $\OPT$ denotes an optimum set of cardinality $k$.
\label{lem:weaksketchguar}	
\end{lemma}

The proof of Lemma~\ref{lem:weaksketchguar} is provided in Appendix~\ref{app:weaksketchguar}. The proofs follows by basic properties of submodular set functions and conditions of the lemma. 

The result in Lemma~\ref{lem:weaksketchguar} tells us that if we can find a minorant set function $\underline{v}$ and a majorant set function $\bar{v}$ that are a $(p,q)$-sketch for a submodular set function $u$ and that satisfy the conditions of the lemma, then any solution of the problem of maximizing the submodular set function $\underline{v}$ subject to a cardinality constraint is a $p/(p+q)$-approximation for the problem of maximizing the submodular set function $u$ subject to the same cardinality constraint. What remains to be done is to find such minorant and majorant set functions, and moreover, show that for every $S$, the value of these functions can be computed in polynomial-time by using only test scores of items in $S$.

We define a minorant set function $\underline{v}$ and a majorant set function $\bar{v}$ which for any given test scores $a_1, a_2, \ldots, a_n$ are defined as, for every $S\subseteq N$,
\begin{equation}
\underline{v}(S) = \min\{a_{i}\mid i\in S\} \hbox{ and } \bar{v}(S) = \max\{a_{i}\mid i\in S\}.
\label{equ:minmax}
\end{equation}

For the minorant set function $\underline{v}$ defined in (\ref{equ:minmax}), the problem of maximizing $\underline{v}(S)$ over $S\subseteq N$ subject to cardinality constraint $|S| = k$ boils down to selecting a set of $k$ items with largest test scores. Obviously, the set functions $\underline{v}$ and $\bar{v}$ defined in (\ref{equ:minmax}) satisfy condition (b) in Lemma~\ref{lem:weaksketchguar}.

We only need to show that there exist test scores $a_1, a_2, \ldots, a_n$ such that $(\underline{v},\bar{v})$ is a $(p,q)$-sketch of the set function $u$. We say that $a_1, a_2, \ldots, a_n$ are $(\clow,\chigh)$-good test scores if $(\underline{v},\bar{v})$ is a $(\clow,\chigh)$-sketch of the set function $u$. If $p/q$ is a constant, we refer to $a_1, a_2, \ldots,a_n$ as \emph{good test scores}. In this case, by Lemma~\ref{lem:weaksketchguar}, selecting a set of $k$ items with largest test scores guarantees a constant-factor approximation for the problem of maximizing the set function $u(S)$ subject to the cardinality constraint $|S|=k$. More generally, we have the following corollary.

\begin{corollary} Suppose that test scores $a_1, a_2, \ldots, a_n$ are $(\clow,\chigh)$-good. Then, greedy selection of items in decreasing order of these test scores yields a utility of value that is at least $p/(p+q)$ times the optimum value. In particular, if $p/q$ is a constant, than the greedy selection guarantees a constant-factor approximation for maximizing the submodular set function $u(S)$ subject to the cardinality constraint $|S|=k$. 
\label{cor:approxtest}
\end{corollary}

We next need to address the question whether for a given stochastic monotone submodular function, there exists good test scores. If good test scores exist, it is possible that there are different definitions of test scores that are good test scores. The lemma shows that replication test scores, defined in Definition~\ref{def:rep}, are good test scores, whenever good test scores exist. 
		
\begin{lemma} Suppose that a utility function has $(\clow,\chigh)$-good test scores. Then, replication scores are $(\clow/\chigh,\chigh/\clow)$-good test scores.
\label{thm:good-test}
\end{lemma}
	
The proof of Lemma~\ref{thm:good-test} is provided in Appendix~\ref{thm:good-test-app}. The lemma tells us to check whether a utility function has good test scores, it suffices to check whether for this utility function, replication test scores are good test scores. If replication test scores are not good test scores for a given utility function, then there exist no good test scores for this utility function. 

In the next lemma, we show that extended diminishing returns, which we introduced in Definition~\ref{def:esa}, is a sufficient condition for replication test scores to be good test scores. 
		
\begin{lemma} Suppose that $g:\reals_+^n \rightarrow \reals_+$ is a symmetric, monotone submodular value function that satisfies the extended diminishing returns property. Then, replication test scores are $(1-1/e,4)$-good test scores, and consequently are good test scores. 
\label{thm:subtest}
\end{lemma}

The proof of Lemma~\ref{thm:subtest} is provided in Appendix~\ref{app:repisgood}. Here we briefly discuss some of the key steps of the proof. First, for the lower bound, we need to show that for every $S \subseteq N$: $u(S) \geq (1-1/e) \underline{v}(S) = (1-1/e)\min\{a_i\mid i\in S\}$, where $a_i$ is the replication test score of item $i$. Suppose that $S = \{1,2,\ldots, k\}$ and without loss of generality, $a_1 = \min\{a_i\mid i\in S\}$. Then, we show by induction that for every $j\in \{1, \ldots, k\}$,
\begin{equation}
\label{eqn_inductionlowerbound}
u(\{1,2,\ldots,j\}) \geq \left(1-\frac{1}{k}\right)u(\{1,2,\ldots,j-1\}) + \frac{1}{k}a_1.
\end{equation}
The proof involves showing that the marginal contribution of adding item $j$ to the set $\{1,2,\ldots,j-1\}$ is closely tied to the marginal contribution of adding item $j$ to a set comprising of $k-1$ other (independently drawn) copies of item $j$. The latter quantity is at most $a_j/k$, which by definition is greater than or equal to $a_1/k$. The exact factor of $1-1/e$ comes from applying the above inequality in a cascading fashion from $u(\{1,2,\ldots,k\})$ to $u(\{1\})$.

The proof of the upper bound is somewhat more intricate. The first step involves carefully constructing a vector $\vec{z} \in \reals^{n-1}_+$ such that $g(\vec{z})$ is larger than $u(S)$ by an appropriate constant factor (say $c$). Imagine that $S^*$ represents some set of $\n-1$ items such that $u(S^*) = g(\vec{z})$. By leveraging monotonicity and submodularity, we have that $u(S) \leq u(S^*) + \sum_{i \in S}(u(S^* \cup \{i\}) - u(S^*))$. Let $\vec{x}$ represent a vector comprising of $k-1$ independent copies of random variables drawn from distribution $F_i$. Now, as per the extended diminishing returns condition, for any realization of $\vec{x}$ such that $g(\vec{x}) \leq g(\vec{z})$, it must be true (assuming that the careful construction $\vec{z}$ leverages Definition~\ref{def:esa}) that:
$$u(S^* \cup \{i\}) - u(S^*)) = g(\vec{z},x_i) - g(\vec{z}) \leq g(\vec{x},x_i) - g(\vec{x}) \hbox{ given that } g(\vec{x}) \leq g(\vec{z}).$$

Moreover, one can apply Markov's inequality to show that $g(\vec{z}) \geq g(\vec{x})$ is true with probability at least $1/c$. Taking the expectation of $\vec{x}$ conditional upon $g(\vec{z}) \geq g(\vec{x})$ gives us the desired upper bound.

The statement of Theorem~\ref{thm:subtest-sum} follows from Corollary~\ref{cor:approxtest} and Lemma~\ref{thm:subtest}.

\section{Submodular Welfare Maximization}
\label{sec:subm-welf-maxim}
	
In this section we present our main result for the stochastic monotone submodular welfare maximization problem. Here, the goal is to find disjoint $S_1, S_2, \ldots, S_m \subseteq N$ satisfying cardinality constraints $|S_j| = k_j$ for all $j\in \{1,2,\ldots,m\}$ that maximize the welfare function $u(S_1, S_2, \ldots, S_m)=\sum_{j=1}^m u_j(S_j)$.

\begin{theorem} \label{thm:welf-maxim} Given an instance of the submodular welfare maximization problem such that the utility functions satisfy the extended diminishing returns property, and the maximum cardinality constraint (i.e., $\max\{k_1,k_2,\ldots,k_m\}$) is $k$, there exists a test score-based algorithm (Algorithm~\ref{alg_greedyassignment}) that achieves a welfare value of at least $1/(24(\log(k)+1))$ times the optimum value. 
\end{theorem}
We briefly comment on the efficacy of test score algorithms for the submodular welfare maximization problem. Unlike the constant factor approximation guarantee obtained in Theorem~\ref{thm:subtest-sum}, test score algorithms only yield a logarithmic-approximation to the optimum for this more general problem. Although constant factor approximation algorithms are known for the submodular welfare maximization problem~\citep{Calinescu2011}, these approaches rely on linear programming and other complex techniques and hence, may not be scalable or amenable to distributed implementation. On the other hand, we focus on an algorithm that is easy to implement in practice but relies on a more restrictive computational model, leading to a worse approximation. Finally, it is worth noting in many actual settings, the value of the cardinality constraint $k$ tends to be rather small in comparison to $n$; e.g., in content recommendation, it is typical to display $25$-$50$ topics per page. In such cases, the loss in approximation due to the logarithmic factor would not be significant.

In the remainder of this section, we provide a proof of Theorem~\ref{thm:welf-maxim}. We will present an algorithm that uses replication test scores, in order to achieve the logarithmic guarantee. The proof is based on using strong sketches of set functions. 


We follow the same general framework as for the submodular function maximization problem, presented in Section~\ref{sec:sbfp}, which in this case amounts to identifying a strong sketch function for each utility set function, defined by using replication test scores, and then using a greedy algorithm for welfare maximization that carefully leverages these replication test scores to achieve the desired approximation guarantee. The following lemma establishes a connection between the submodular welfare maximization problem and strong sketches. 
	 	
\begin{lemma}
Consider an instance of the submodular welfare maximization with utility functions $u_1, u_2, \ldots, u_m$ and parameters of the cardinality constraints $k_1,k_2, \ldots, k_m$. Let $\vec{\OPT} = (\OPT_1, \OPT_2, \ldots, \OPT_m)$ denote an optimum partition of items. Suppose that for each $j\in M$, $(v_j,v_j)$ is a $(\clow,\chigh)$-sketch of $u_j$, and that $S_1, S_2,\ldots, S_m$ is an $\alpha$-approximation to the welfare maximization problem with utility functions $v_1, v_2,\ldots, v_m$ and the same cardinality constraints. Then,
$$
\sum_{j=1}^m u_j(S_j) \geq \alpha\frac{\clow}{\chigh}u(\vec{\OPT}) =\alpha\frac{\clow}{\chigh}\sum_{j=1}^m u_j(\OPT_j) .
$$
\label{clm_strskeproject}
\end{lemma}

The proof of Lemma~\ref{clm_strskeproject} is provided in Appendix~\ref{app:strsketch} 

We next define set functions that we will show to be strong sketch for utility functions of the welfare maximization problem that satisfy the extended diminishing returns property. Fix an arbitrary set $S\subseteq N$ such that $|S| = k$ and $j\in M$. Let $a_{i,j}^r$ be the replication score of item $i$ for value function $g_j$ and cardinality parameter $r$, i.e., 
$$
a^r_{i,j} = \E[g_j(X_i^{(1)}, X_i^{(2)},\ldots, X_i^{(r )},\phi,\ldots,\phi)].
$$
Let $\pi(S,j) = (\pi_1(S,j), \dots, \pi_k(S,j))$ be a permutation of items in $S$ defined as follows: 
\begin{equation}
\begin{array}{rl}
\pi_1(S,j) &= \argmax_{i \in S}a_{i,j}^1\\
\pi_2(S,j) &= \argmax_{i \in S\setminus\{\pi_1(S,j)\}}a_{i,j}^2\\
&\vdots  \\
\pi_k(S,j) &= \argmax_{i\in S\setminus \{\pi_1(S,j),\ldots,\pi_{k-1}(S,j)\}} a_{i,j}^{k}.
\end{array}
\label{equ:pi}
\end{equation}
We define a set function $v_j:2^N\rightarrow \reals_+^n$ for every set $S\subseteq N$ of cardinality $k$ as follows:
\begin{equation}
v_j(S) = a_{\pi_1(S,j),j}^1 + \frac{1}{2}a_{\pi_2(S,j),j}^2 + \cdots + \frac{1}{k}a_{\pi_k(S,j),j}^k.
\label{vs}
\end{equation}

The definition of set function $v_j$ in (\ref{vs}) can be interpreted as defining the value $v_j(S)$ for every given set $S$ to be additive with coefficients associated with items corresponding to their virtual marginal values in a greedy ordering of items with respect to these virtual marginal values. 


Given a partition of items in disjoint sets $S_1, S_2, \ldots, S_m$, we define a welfare function $v(S_1, S_2,\ldots,S_m) = \sum_{j=1}^m v_j(S_j)$. We next show that that set functions defined in (\ref{vs}) are strong sketch functions. 

\begin{lemma} Suppose that a set function $u_j$ is defined as the expected value of a symmetric, monotone submodular value function that satisfies the extended diminishing returns condition. Then, the set function $v_j$ given by (\ref{vs}) is a $(1/(2(\log(k)+1)),6)$ strong sketch of $u_j$.  
\label{thm:SMBsketch}
\end{lemma}

The proof of Lemma~\ref{thm:SMBsketch} is provided in Appendix~\ref{app_proof_lemmasmbsketch}. 

By Lemma \ref{clm_strskeproject} and Lemma \ref{thm:SMBsketch}, for any stochastic monotone submodular welfare maximization problem with utility functions satisfying the extended diminishing returns condition, any $\alpha$-approximate solution to the submodular welfare maximization problem, we refer to as a surrogate welfare maximization problem with the welfare function $v(S_1, S_2, \ldots, S_m)$ subject to the same cardinality constraints as in the original welfare maximization problem, is a $c\alpha/(\log(k)+1)$-approximate solution to the original welfare maximization problem, where $c$ is a positive constant. It remains to now to show that the surrogate welfare maximization problem admits an $\alpha$-approximate solution. We next show that a naural greedy algorithms applied to the surrogate welfare maximization problem guarantees a $1/2$-approximation for this problem. 

Consider a natural greedy algorithm for the surrogate welfare maximization problem that works for the case of one or more partitions. Given the replication test scores for all items and all partitions, in each step $r$, the algorithm adds an unassigned item $i$ and partition $j$ that maximizes $a^{r_j}_{i,j}$ where $r_j$ is the number of elements assigned to partition $j$ in previous steps. That is, in each iteration, an assignment of an item to a partition is made that yields the largest marginal increment of the surrogate welfare function. The algorithm is more precisely defined in Algorithm~\ref{alg_greedyassignment}.

\begin{algorithm}[t]
	Initialize assignment $S_1 = S_2 = \ldots = S_m = \emptyset$\, $A=\{1,2,\dots,n\}$, $P=\{1,2,\dots,m \}$\\
	/* $S_j$ and $A$ denote the set of assigned items to partition $j$ and the set of unassigned items */ \\
	\While {$|A| >0$ and $|P|>0$}
	{$(i^*, j^*) = \argmax_{(i,j)\in A\times P } a^{|S_j|+1}_{i,j}/(|S_j|+1)$ ~~ /* with random tie break */\\ 
		$S_{j^*} \leftarrow S_{j^*} \cup \{i^* \}$ and $A \leftarrow A\setminus \{i^*\}$ ~~ /* assign item $i^*$ to partition $j^*$ */\\
		\If{$|S_{j^*}| \ge k_j$}
		{$P \leftarrow P \setminus \{j^* \}$ ~~ /* remove partition $j^*$ from the list */}
		
	}
	\caption{Greedy Algorithm for Submodular Welfare Maximization Problem}
	\label{alg_greedyassignment}
\end{algorithm}

In the following lemma, we show that the greedy algorithm guarantees a $1/2$-approximation for the surrogate welfare maximization problem. 

\begin{lemma} The greedy algorithm defined by Algorithm~\ref{alg_greedyassignment} outputs a solution that is a $\frac{1}{2}$-approximation for the submodular welfare maximization problem of maximizing $v(S_1, S_2, \ldots, S_m)$ over partitions of items $(S_1, S_2, \ldots, S_m)$ that satisfy cardinality constraints. 
\label{lem:2}
\end{lemma}

The proof of Lemma~\ref{lem:2} can be found in Appendix~\ref{alg_greedyassignment}. The proof is similar in spirit to that of the $\frac{1}{2}$-approximate greedy algorithm for submodular welfare maximization proposed by~\cite{lehmann2006combinatorial}. Unfortunately, one cannot directly utilize the arguments in that paper since the sketch function that we seek to optimize---$v_j(S_j)$---may not be submodular. Instead, we present a novel montonicity argument and leverage it to provide the following upper and lower bounds: $v_j(S_j) \geq v_j(S_j \setminus \{\pi_r(S_j, j)\}) \geq v_j(S_j) - \frac{a^r_{\pi_r(S_j, j),j}}{|r|}$ for all $S_j \subseteq N$ and $1 \leq r \leq k_j$. Finally, we apply these bounds in a cascading manner to show the desired $\frac{1}{2}$-approximation factor claimed in Lemma~\ref{lem:2}.

\section{Discussion and Additional Results}
\label{sec:disc}

In this section we first illustrate the use of test scores and discuss numerical results for the simple example introduced in Section~\ref{sec:intro}. We then discuss performance of simple test scores, namely mean and quantile scores, and characterize their performance for the constant elasticity of substitution value function. Finally, we discuss why for the stochastic monotone submodular welfare maximization problem we have to use different sketch functions than those we used for the stochastic monotone submodular function maximization problem. 

\subsection{Numerical Results for a Simple Illustrative Example}
\label{sec:disc1}

We consider the example of two types of items that we introduced in Section~\ref{sec:intro}. Recall, in this example the ground set of items is partitioned in two sets $A$ and $B$ with set $A$ comprising of safe items whose each individual performance is of value $a$ with probability $1$ and set $B$ comprising of risky items whose each individual performance is of value $b/p$ with probability $p$, and value $0$, otherwise. Here $a$, $b$, and $p$ are parameters such that $a,b>0$ and $p\in (0,1]$. We assume that $b \geq a$ and $|A|, |B|\geq k$. The value function is assumed to be the best-shot value function. 

We say that a set $S$ of items of cardinality $k$ is of type $r$ if it contains exactly $r$ risky items for $r = 0, 1,\ldots, k$. For each $r\in \{0,1\ldots, k\}$, let $S_r$ denote an arbitrary type-$r$ set. The realized value of set $S_r$ is $b/p$ if at least one risky item in $S_r$ achieves value $b/p$ and is equal to $a$, otherwise. Hence, we have
$$
u(S_r) = a(1-p)^r + \frac{b}{p}(1-(1-p)^r).
$$
Notice that the value of $u(S_r)$ monotonically increases in $r$, hence it is optimal to select a set that comprises of $k$ risky items, i.e. a set of type $k$. 

We consider sample-average replication test scores, which for a given number of samples per item replica $T\geq 1 $, are defined as
$$
\hat{a}_i = \frac{1}{T}\sum_{t=1}^T \max\{X_i^{(t,1)}, X_i^{(t,2)}, \ldots, X_i^{(t,k)}\}
$$
where $X_i^{(t,j)}$ are independent samples over $i$, $t$ and $j$ with $X_i^{(t,j)}$ sampled from distribution $P_i$. The output of the test score algorithm consists of a set of $k$ items with highest sample-average replication test scores. The output results in an error if, and only if, it contains at least one safe item. We evaluate the probability of error of the algorithm by running the test score algorithm for a number of repeated experiments. 

\begin{figure}[t]
\begin{center}
\includegraphics[width=0.45\linewidth]{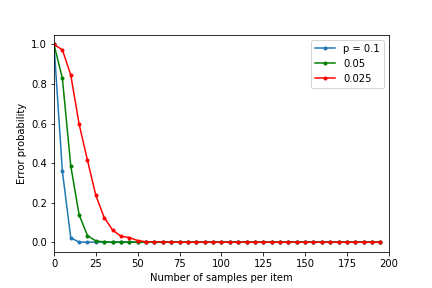}
\includegraphics[width=0.45\linewidth]{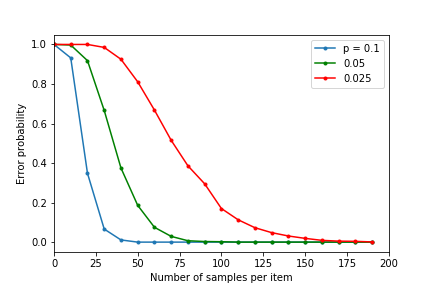}
\\
\caption{Probability of error of the test score algorithm versus the number of samples per item for (left) $k = 5$ and (b) $k = 10$, in each case for different values of parameter $p = 0.025, 0.05$ and $0.1$. Other parameters are set as $|A| = |B| = 10$, $a=1$ and $b=2$. The results are for the number repeated experiments equal to $1000$. }
\label{fig:ts_vs_T}
\end{center}
\end{figure}

In Figure~\ref{fig:ts_vs_T}, we show the probability of error versus the number of samples per item, for different values of parameters $k$ and $p$. Notice that the number of samples per item is equal to $Tk$ where $T$ is the number of samples per item replica. We observe that (a) the probability of error decreases with the number of samples per item, (b) the probability of error is larger for larger set size, and (c) the number of samples per item required to achieve a fixed value of probability of error increases with the risk of item values, i.e. for smaller values of parameter $p$. In Figure~\ref{fig:ts_vs_p}, we show the probability of error versus the value of parameter $p$, for different values of parameters $k$ and $T$. This further illustrates that a larger number of samples is needed to achieve a given probability of error the later the risk of items. In fact, one can show that a sufficient number of samples per item is $O((k/p^2)\log(n/\delta))$ to guarantee that the probability of error is of value at most $\delta$; we provide details in Appendix.  

\begin{figure}[t]
\begin{center}
\includegraphics[width=0.45\linewidth]{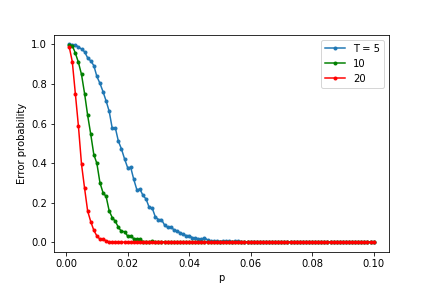}
\includegraphics[width=0.45\linewidth]{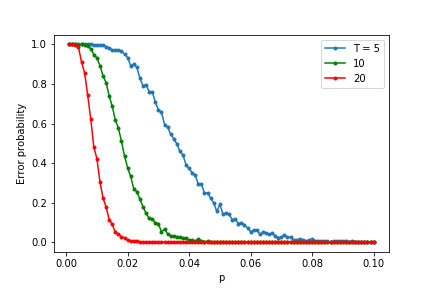}
\\
\caption{Probability of error of the test score algorithm versus the value of parameter $p$ for (left) $k = 5$ and (right) $k = 10$, in each case for different number of samples per item replica $T=5$, $10$ and $20$. Other parameters are set as given in the caption of Figure~\ref{fig:ts_vs_T}. }
\label{fig:ts_vs_p}
\end{center}
\end{figure}

We further consider a sample averaging method that amounts to enumerating feasible sets of items, for each feasible set $S$ of items estimating the value of $u(S)$, and selecting a set with largest estimated value. The value of $u(S)$ is estimated by the estimator defined as
$$
\hat{u}(S) = \frac{1}{T}\sum_{t=1}^T \max\{X_i^{(t)}\mid i\in S\}
$$
where $X_i^{(t)}$ are independent samples over $i$ and $t$ with $X_i^{(t)}$ sampled from distribution $P_i$.

In Figure~\ref{fig:saa_vs_T} we show the probability of error versus the number of samples per item for the test score algorithm and the sample averaging approach (SAA). We observe that the probability of error is larger for the SAA method. Intuitively, this happens because the SAA method amounts to comparison of all possible sets of items of different types, while the test score method for replication test scores amounts to comparison of sets that consists of either all safe or all risky items. The SAA method is computationally expensive as it requires enumeration of $\binom{n}{k}$ sets of items, which is prohibitive in all cases but for small values of parameter $k$. For the example under consideration, the number of samples per item needed to guarantee a given probability of error can be analytically characterized; we show this in Appendix. The sufficient number of samples per item scales in the same as way as for the test score algorithm, for fixed value of $k$ and asymptotically small values of parameter $p$, but for a fixed value of $p$ increases exponentially in parameter $k$.  

\begin{figure}[t]
\begin{center}
\includegraphics[width=0.35\linewidth]{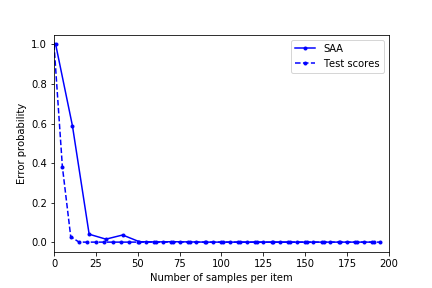}\hspace*{-5mm}
\includegraphics[width=0.35\linewidth]{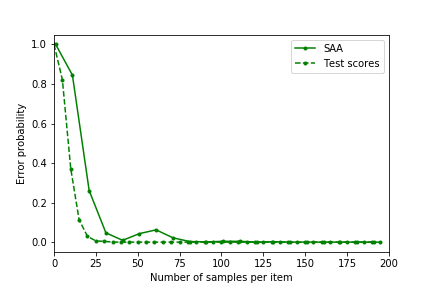}\hspace*{-5mm}
\includegraphics[width=0.35\linewidth]{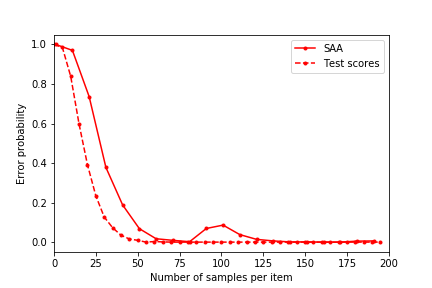}
\caption{Probability of error versus the number of samples per item for SAA and test score algorithms, for (left) $p=0.1$, (middle) $p=0.05$, and (right) $p=0.025$. The setting of other parameters is as in Figure~\ref{fig:ts_vs_T} for $k = 5$.}
\label{fig:saa_vs_T}
\end{center}
\end{figure}

In summary, our numerical results demonstrate the efficiency of the test score algorithm for different values of parameters and in comparison with the sample averaging approach.

\subsection{Mean and Quantile Test Scores}
\label{sec:mean}

As we already mentioned, the mean test scores are defined as expected values $a_i = \E[X_i]$. The quantile test scores are defined as $a_i = \E[X_i\mid P_i(X_i)\geq \theta]$, for a parameter $\theta \in [0,1]$. For the value of parameter $\theta = 0$, the quantile test score corresponds to mean test score. In general, the quantile test score is the expected individual performance of an item conditional on it being larger than a threshold value. 

Neither mean test scores nor quantile test scores can guarantee a constant-factor approximation for the submodular function maximization problem. We demonstrate this by two simple examples that convey intuitive explanations on why these test scores can fail to provide desired guarantee. We then present tight approximation bounds for the CES utility functions.

\noindent {\bf Example 1 (mean test scores):} Suppose that the utility is according to the best-shot function and that the selection is greedy using mean test scores. Suppose that there are two types of items: (a) deterministic performance items whose each individual performance is of value $1$ with probability $1$ and (b) random performance items whose individual performances are independent with expected value strictly smaller than $1$ and a strictly positive probability of being larger than $1$. Then, the algorithm will select all items to be those with deterministic performance. This is clearly suboptimal under the best-shot production where having selected an item with deterministic performance, the only way to increase the performance of a set of items with some probability is to select an item with random performance. Such an instance can be chosen such that the algorithm yields the utility that is only factor $O(1/k)$ of the optimum value.

\noindent {\bf Example 2 (quantile test scores):} Suppose that the utility function is the sum of individual performances and consider greedy selection with respect to quantile test scores with threshold parameters $\theta_{i} = 1-1/k$. Suppose there are two types of items: (a) deterministic performance items whose each individual performance is of value $1$ with probability $1$ and (b) random performance items whose individual performances are independent of value $a > 1$ with probability $p > 1/k$ and otherwise equal to zero. For random performance items, the mean test score is of value $a p$ and the quantile test score is of value $a$. The algorithm will choose all items to be random performance items, which yields the utility of value $k a p$. On the other hand, choosing items that output deterministic performance, yields the utility of value $k$. Since $a$ and $p$ can be chosen to be arbitrarily near to values $1$ and $1/k$, respectively, we observe that the algorithm yields the utility that is $O(1/k)$ of the optimum value.  

We next present a tight approximation bound for the CES utility function with parameter $r\geq 1$. Recall that the CES utility production provides an interpolation between two extreme cases: a linear function (for $r = 1)$ and the best-shot function (for the limit as $r$ goes to infinity). Intuitively, we would expect that greedy selection with respect to mean test scores would perform well for small values of parameter $r$, but that the approximation would get worse by increasing parameter $r$. The following theorem makes this intuition precise. 

\begin{proposition}[mean test scores] Suppose that the utility function $u$ is according to the CES production function with parameter $r\geq 1$. For given cardinality parameter $k\geq 1$, let $M$ be a set of $k$ items in $N$ with highest mean test scores. Then, we have 
$$
u(M) \geq \frac{1}{k^{1-1/r}}u(\OPT).
$$  
Moreover, this bound is tight.
\label{thm:cesmeantestscores}
\end{proposition}

The proof of Proposition~\ref{thm:cesmeantestscores} is provided in Appendix~\ref{sec:cesmeantestscores}. The proposition shows how the approximation factor decreases with the value of parameter $r$. In the limit of asymptotically large $r$, the approximation factor goes to $1/k$. This coincides with the approximation factor obtained for the best-shot function in \cite{kleinberg2015team}.

Intuitively, we would expect that quantile test scores would yield a good approximation guarantee for the CES utility function with large enough parameter $r$. This is because we know that for the best-shot utility function, the quantile test scores can guarantee a constant-factor approximation, which was established in \cite{kleinberg2015team}. The following theorems makes this intuition precise.
		
\begin{proposition}[quantile test scores] Suppose that the utility is according to the CES production function with parameter $r$ and that the selection is greedy using quantile test scores with $\theta = 1-c/k$ and $c > 0$. Then, we have
\begin{enumerate}
\item[(a)] if $r = o(\log(k) )$ and $r>1$, the quantile test scores cannot guarantee a constant-factor approximation for any value of parameter $c > 0$;
\item[(b)] if $r = \Omega(\log(k))$, the quantile test scores with $c = 1$ guarantee a constant-factor approximation.
\end{enumerate}
\label{thm:quant}
\end{proposition}

The proof of Proposition~\ref{thm:quant} is provided in Appendix~\ref{sec:quant}. The proposition establishes that quantile test scores can guarantee a constant-factor approximation if, and only if, the parameter $r$ is larger than a threshold whose value is $\Theta(\log(k))$.

\subsection{Sketch Functions used for the Welfare Maximization Problem}

In Section~\ref{sec:subm-welf-maxim} we established an approximation guarantee for the stochastic monotone submodular welfare maximization problem using the concept of strong sketches of set functions. This is in contrast to Section~\ref{sec:SFM} where used non-strong sketches for the submodular function maximization problem. One may wonder whether we could have used the theory of good test scores developed for submodular function maximization for the more general problem of submodular welfare maximization. Specifically, given an instance, one may have used the characterization in Definition~\ref{def:weaksketch} to maximize either $\underline{v}(S_1, S_2, \ldots, S_m) = \sum_{j=1}^m \underline{v}_j(S_j)$ or $\bar{v}(S_1, S_2, \ldots, S_m) =\sum_{j=1}^m \bar{v}_j(S_j)$, with $\underline{v}_j$ and $\bar{v}_j$ as defined in (\ref{equ:minmax}), over all feasible assignments. However, as we show next, such approaches can lead to highly sub-optimal assignments even for simple instances. 

\paragraph{Example 1:} Suppose we use an algorithm for maximizing the welfare function $\underline{v}(S_1, S_2,\ldots, S_m)$ subject to cardinality constraints. 

Consider a problem instance with $n = r^2$ items and $m=r$ partitions with each partition having a cardinality constraint with $k_j = r$ for all $r$. All items are assumed to exhibit deterministic performance: $r$ items (referred to as heavy items) have performance of value $1$, i.e., $X_i = 1$ with probability $1$, while the remaining items have performance of zero value. Assume that value functions are best-short functions $g_j(S) = \max\{x_i\mid i\in S\}$ for each partition $j$. 
		
The optimum solution for the given problem instance is when each of the heavy items is assigned to a different partition, leading to the welfare of value $r$. On the contrary, the algorithm assigns all heavy items to same partition, which yields a welfare of value $1$. Hence, the algorithm achieves the welfare that is $1/\sqrt{n}$ factor of the optimum, which can be made arbitrarily small by choosing large enough number of items $n$. 
		
\paragraph{Example 2:} Suppose now that we use an algorithm for maximizing the welfare function $\bar{v}(S_1, S_2, \ldots, S_m)$ subject to cardinality constraints. 

Consider a problem instance with $n = 2r$ items and $m = r+1$ partitions, where partition $1$ has a cardinality constraint with $k_1 = r$, and each partition $1 < j\leq m$ has $k_j = 1$. All items are assumed to have deterministic performance once again: one heavy item with performance of value $\sqrt{r}$,  $r-1$ medium items with performance of value of $1$, and, finally, the remaining items with zero-valued performance. Assume that value functions are $g_1(\vec{x}) = \sum_{i=1}^n x_i$ and $g_j(\vec{x}) = (1/\sqrt{r})\max\{x_i\mid i=1,2,\ldots,n\}$ for partitions $1 < j \leq m$. 
			
The optimum solution assigns all items to partition $1$, which yields a welfare of value $r+\sqrt{r}-1$, whereas the algorithm assigns the heavy item to partition $1$ and the medium items spread across, which yields a welfare of value less than $2\sqrt{r}$.	Hence, the algorithm achieves the welfare which is less than $2\sqrt{2}/\sqrt{n}$ of the optimum welfare, which can be made arbitrarily small by taking large enough number of items $n$. 
\section{Conclusion}
\label{sec:conc}

In this work, we presented a new algorithmic approach for the canonical problem of (stochastic) submodular maximization known as test score algorithms. These algorithms are particularly appealing due to their simplicity and natural interpretation as their decisions are contingent only on individual item scores that are computed based on the distribution that captures the uncertainty in the respective item's performance. Although test score based methods have been studied in an ad-hoc manner in previous literature~\citep{kleinberg2015team}, our work presents the first systematic framework for solving a broad class of stochastic combinatorial optimization problems by approximating complex set functions using simpler test score based sketch functions. By leveraging this framework, we show that it is possible to obtain good approximations under a natural (extended) diminishing returns property, namely: $(i)$ a constant factor approximation for the problem of maximizing a stochastic submodular function subject to a cardinality constraint, and $(ii)$ a logarithmic-approximation guarantee for the more general stochastic submodular welfare maximization problem. It is worth noting that since test score algorithms represent a more restrictive computational model, the guarantees obtained in this paper are not as good as those of the best known algorithms for both these problems. However, test score based approaches provide three key advantages over more traditional algorithms that make them highly desirable in practical situations relating to online platforms:

\begin{itemize}
\item \emph{Scalability}: The test score of an item depends only on its own performance distribution. Therefore, when new items are added or existing items are removed from the ground set, this does not alter the scores of any other items. Since our algorithm selects items with the highest test scores, its output would only require simple swaps when the ground set changes. 

\item \emph{Distributed Implementation}: Test score algorithms can be easily parallelized as the test score of an item can be computed independently of the performance distribution of other items. Moreover, the final algorithm itself involves a trivial greedy selection and does not require any complex communication between parallel machines. 

\item \emph{Fewer Oracle Calls}: Test score algorithms only query the value of the function $E[g(\vec{x})]$ once per item---$n$ oracle calls in total---which is an order of magnitude smaller than the number required by traditional approaches. Moreover, these oracle calls are simple in that they do not require drawing samples from the distributions of multiple items, which may be expensive.  
\end{itemize}


Future work may consider lower bounds for test score-based algorithms for different sub-classes of monotone stochastic submodular set functions. In particular, it would be of interest to consider instances of set functions that do not belong to the class of set functions identified in this paper. It is also of interest to consider tightness (inapproximability) of approximation factors. Finally, it would also be of interest to study approximation guarantees when using statistical estimators for test scores, and not expected values as in this paper.

%

%

\bibliographystyle{informs2014}

\bibliography{reference}

\begin{APPENDIX}{Proofs and Additional Results}

\section{Validation of the Extended Diminishing Returns Property}
\label{sec:check}
		
It is easy to verify that all value functions defined in Section~\ref{sec:val} are such that their expected values are non-negative, monotone submodular set functions. We next show that all these value functions also satisfy the extended diminishing returns condition, formally defined in Definition~\ref{def:esa}. 
		
We need to check that a value function $g$ is such that whenever for given $v\in \reals_+$ there exists $\vec{y}'\in \reals_+^d$ such that $g(\vec{y}')=v$, then there exists $\vec{y}=(y_1,\ldots,y_d)^\top \in \reals_+^d$ such that $g(\vec{y}) = v$ and for all $\vec{x} = (x_1, \ldots, x_d)^\top\in \reals_+^d$ such that $g(\vec{x}) \leq g(\vec{y})$, it holds
\begin{equation}
g(x_1, \ldots, x_d, z) - g(x_1, \ldots, x_d) \geq g(y_1,\ldots,y_d,z) - g(y_1, \ldots, y_d), \hbox{ for all } z\in \reals_+.
\label{equ:appcond1}
\end{equation}
We first prove that for all of the functions defined in Section~\ref{sec:val} except for top-$r$ with $r > 1$ satisfy a stronger version of the above condition which is true for all points $\vec{y}\in \reals_+^d$ such that $g(\vec{y}) = v$. According to the stronger condition, for every $\vec{x}, \vec{y}\in \reals_+^d$ such that $g(\vec{x}) \leq g(\vec{y})$, it holds:
\begin{equation}
g(x_1, \ldots, x_d, z) - g(x_1, \ldots, x_d) \geq g(y_1,\ldots,y_d,z) - g(y_1, \ldots, y_d), \hbox{ for all } z\in \reals_+.
\label{equ:appcond2}
\end{equation}

We begin by proving that all of the functions defined in Section~\ref{sec:val} except top-$r$ satisfy the stronger condition as per~\eqref{equ:appcond2}.

		\paragraph{Total production: $g(\vec{x})=\bar{g}(\sum_{i=1}^n x_i)$} In this case, $g(\vec{x}) \leq g(\vec{y})$ is equivalent to $\sum_{i=1}^d x_i \leq \sum_{i=1}^d y_i$ and (\ref{equ:appcond2}) is equivalent to 
		$$
		\bar{g}\left(\sum_{i=1}^d x_i + z\right) - \bar{g}\left(\sum_{i=1}^d x_i\right) \geq \bar{g}\left(\sum_{i=1}^d y_i + z\right) - \bar{g}\left(\sum_{i=1}^d y_i\right), \hbox{ for all } z\in \reals_+.
		$$

Let $x = \sum_{i=1}^d x_i$ and $y = \sum_{i=1}^d y_i$. With this new notation, the extended diminishing returns condition is equivalent to saying that for all $x,y\in \reals_+$ such that $x\leq y$, 
$$
\bar{g}(x+z)-\bar{g}(x) \geq \bar{g}(y+z)-\bar{g}(y), \hbox{ for all } z\in \reals_+
$$
which obviously holds true because $\bar{g}$ is assumed to be a monotone increasing and concave function.
		
		\paragraph{Best-shot: $g(\vec{x}) = \max\{x_1,x_2,\ldots,x_n\}$.} In this case, $g(\vec{x}) \leq g(\vec{y})$ is equivalent to
$$
\max\{x_1,\ldots,x_d\}\leq \max\{y_1,\ldots,y_d\}
$$
and (\ref{equ:appcond2}) is equivalent to
		$$
		\max\{x_1,\ldots,x_d,z\} - \max\{x_1,\ldots,x_d\} \geq \max\{y_1,\ldots,y_d,z\} - \max\{y_1,\ldots,y_d\} \hbox{ for all } z\in \reals_+.
		$$
		
We consider three different cases. 

\begin{itemize}
\item Case 1: $z\geq \max\{y_1,\ldots,y_d\}$. In this case, $\max\{\vec{x},z\} - \max\{\vec{x}\} = z - \max\{\vec{x}\} \geq z - \max\{\vec{y}\} = \max\{\vec{y},z\} - \max\{\vec{y}\}$. Hence, extended diminishing returns holds. 

\item Case 2: $\max\{x_1,\ldots, x_d\} \leq z < \max\{y_1,\ldots,y_d\}$. In this case, condition (\ref{equ:appcond2}) is equivalent to $z\geq \max\{x_1,\ldots,x_d\}$, which holds by assumption. 

\item Case 3: $z < \max\{x_1,\ldots,x_d\}$. In this case, condition (\ref{equ:appcond2}) is equivalent to $0 \geq 0$ and thus trivially holds. 
\end{itemize}

\paragraph{CES: $g(\vec{x}) = (\sum_{i=1}^n x_i^r)^{1/r}$, for parameter $r\geq 1$.} Let $x = \sum_{i=1}^d x_i^r$, $y = \sum_{i=1}^r y_i^r$ and $w = z^p$. Condition (\ref{equ:appcond2}) is equivalent to
		$$
		(x+w)^{1/r} - x^{1/r} \geq (y+w)^{1/r} - y^{1/r}
		$$
while $g(\vec{x}) \leq g(\vec{y})$ is equivalent to $x \leq y$. Since $r \geq 1$, the function $f(x) = x^{1/r}$ is an increasing concave function. Hence, it follows that condition (\ref{equ:appcond2}) holds as long as $g(\vec{x}) \leq g(\vec{y})$.
		 
\paragraph{Success-probability: $g(\vec{x}) = 1-\prod_{i=1}^n(1-p(x_i))$} By a simple algebra, condition (\ref{equ:appcond2}) is equivalent to
		$$
		\prod_{i=1}^d p(x_i)(1-p(z)) \geq \prod_{i=1}^d p(y_i)(1-p(z))
		$$
	while $g(\vec{x}) \leq g(\vec{y})$ is equivalent to
		$$
		\prod_{i=1}^d p(x_i)\geq \prod_{i=1}^d p(y_i).
		$$
Hence, condition (\ref{equ:appcond2}) holds as long as $g(\vec{x}) \leq g(\vec{y})$.

Finally, we prove that the top-$r$ function satisfies~\eqref{equ:appcond1} for $r > 1$. Recall that when $r=1$, top-$r$ coincides with the best-shot function, for which we already showed that the extended diminishing returns condition holds. 
	
\paragraph{Top-$r$: $g(\vec{x}) = \sum_{i=1}^r x_{(i)}$, where $x_{(i)}$ is the $i$--th largest element in $\vec{x}$.} Fix $v\in \reals_+$. Without loss of generality, suppose that$d \geq r$ and define $\vec{y} = (y_1, \ldots, y_d)^\top \in \reals^d$ such that $y_j = v/r$ for $1 \leq j \leq r$ and $y_j = 0$ for all $r < j \leq d$.\footnote{The proof when $r < d$ is trivial because $g(\vec{x},z) - g(\vec{x}) = g(\vec{y},z) - g(\vec{y})=z$.} Clearly, $g(\vec{y}) = v$.

Let $\vec{x} \in \reals_+^d$ be any point such that $g(\vec{x}) \leq g(\vec{y})$.  We prove~\eqref{equ:appcond1} for the following two different cases:
\begin{itemize}
\item Case 1: $z\geq v/r$: In this case, $g(\vec{y},z) - g(\vec{y}) = z - v/r$. Since $g(\vec{x}) \leq g(\vec{y})$, it must be the case that the $r$-th largest element in $\vec{x}$, i.e. $x_{(r)}$, is smaller than or equal to $g(\vec{y})/r = v/r$. Thus, we have that $g(\vec{x},z) - g(\vec{x}) = z - x_{(r)} \geq z - g(\vec{y})/r = g(\vec{y},z) - g(\vec{y})$ and so, the claim follows.
\item Case 2: $z \leq v/r$: The claim trivially follows in this case because $g(\vec{y},z) = g(\vec{y})$ and so,  $g(\vec{y},z) - g(\vec{y}) = 0$, whereas $g(\vec{x},z) - g(\vec{x}) \geq 0$.
\end{itemize}

\section{Proof of Lemma~\ref{lem:weaksketchguar}}
\label{app:weaksketchguar}
	
We first note the following inequalities
		
		$$
		u(\OPT) \leq u(S^*) + u(\OPT\setminus S^*) \leq u(S^*) + \chigh\bar{v}(\OPT \setminus S^*).
		$$
		
		The first inequality comes from the fact that all submodular functions are subadditive, i.e. for any submodular set function $u$, it holds $u(A \cup B) \leq u(A) + u(B)$. The second inequality comes from the sketch upper bound. 
		
		Now, consider any set $T$ of cardinality $k$ such that $\OPT \setminus S^* \subseteq T$ that is disjoint from $S^*$, i.e. $S^*\cap T = \emptyset$. By the condition of the lemma, we have that $\bar{v}(T) \leq \underline{v}(S^*)$ and $p\underline{v}(S^*) \leq u(S^*)$. Therefore, we have
		
		$$u(\OPT) \leq u(S^*) + \chigh\underline{v}(S^*) \leq u(S^*) + \frac{\chigh}{\clow}u(S^*)$$
		
		which completes the proof.
		
\section{Proof of Lemma~\ref{thm:good-test}}
\label{thm:good-test-app}

Suppose that a set function $u$ has $(\clow,\chigh)$-good test scores $a_1,a_2,\ldots, a_n$, i.e. for every $S\subseteq N$ such that $|S| = k$,
		\begin{equation}
		\clow \min\{a_i\mid i\in S\} \leq u(S) \leq \chigh \max\{a_i\mid i\in S \}. 
		\label{eq:t1}
		\end{equation}
		
		Let $r_1,r_2,\ldots,r_n$ be replication test scores, i.e.\footnote{Hereinafter, we slightly abuse the notation by writing $u(S)$ for a set of item $i$ replicas $S = \{i^{(1)}, \ldots, i^{(k)}\}$ while $u$ is defined as a set function over $2^N$. A proper definition would extend the definition of $u$ over $2^{\tilde{N}}$ where $\tilde{N}$ includes $n$ instances of each item $i\in N$ but this would be at the expense of more complex notation.} 
		\begin{equation}
		r_i = \E[g(X_i^{(1)},\dots,X_i^{(k)},\phi,\ldots,\phi)] = u(\{i^{(1)},\dots,i^{(k)}\})
		\label{eq:rpl}
		\end{equation}
		where $X_i^{(1)}, X_i^{(2)}, \ldots, X_i^{(k)}$ are independent random variables with distribution $P_i$ and $i^{(1)}, i^{(2)},\ldots,i^{(k)}$ are independent replicas of item $i$.
		
		By assumption, $a_1,a_2,\ldots,a_n$ are $(\clow,\chigh)$-good test scores, hence 
		\begin{equation}
		\clow a_i \leq u(\{i^{(1)},\dots,i^{(k)}\}) \leq \chigh a_i.
		\label{eq:t2}
		\end{equation}
		
		From \eqref{eq:t1}, \eqref{eq:rpl}, and \eqref{eq:t2}, we have that for every $S\subseteq N$ such that $|S| = k$,
		$$
		\frac{\clow}{\chigh} \min\{r_i\mid i\in S\} \leq \clow \min\{a_i\mid i\in S\} \leq u(S) \leq \chigh \max\{a_i\mid i\in S\} \leq \frac{\chigh}{\clow}\max\{r_i\mid i\in S\}
		$$
		which implies that replication test scores are $(p/q,q/p)$-good test scores.

\section{Proof of Lemma~\ref{thm:subtest}}
\label{app:repisgood}

We first prove the lower bound and and then the upper bound as follows.
		
\paragraph{Proof of the lower bound.} Without loss of generality, let us consider the set $S = \{1,2,\ldots,k\}$ and assume that $a_{1} = \min\{a_{i}\mid i\in S\}$. We claim that 
\begin{equation}
u(\{1,\ldots,j\}) \geq \left(1-\frac{1}{k} \right)u(\{1,\dots,j-1 \}) + \frac{1}{k}a_{1} \hbox{ for all } j\in \{1,2,\ldots,k\}.
\label{equ:claimu}
\end{equation} 
		
From this, we can use a cascading argument to show that $u(S) \geq (1-(1-\frac{1}{k})^k) a_{1} \geq (1-\frac{1}{e})a_{1}$. 

We begin by proving the claim by (\ref{equ:claimu}). For $j = 1$, since $u$ is a non-negative, monotone submodular set function, we have
		\begin{equation}
		u(\{ 1\}) = \frac{1}{k}\sum_{t=1}^k u(\{ 1^{(t)}\})  \ge \frac{1}{k} u(\{1^{(1)},\dots,1^{(k)} \}) = \frac{1}{k}a_{1}. 
		\label{eq:3}
		\end{equation}
For $j > 1$, we have
		\begin{eqnarray}
		u(\{1,\dots,j \}) & = &  u(\{1,\dots,j-1 \}) + [u(\{1,\dots,j \})-u(\{1,\dots,j-1 \})]\nonumber\\
		& \stackrel{(a)}{\ge} & u(\{1,\dots,j-1 \}) +\frac{1}{k} [u(\{1,\dots,j-1,j^{(1)},\dots,j^{(k)} \})-u(\{1,\dots,j-1 \})] \nonumber\\
		& \stackrel{(b)}{\ge} & u(\{1,\dots,j-1 \}) +\frac{1}{k}[u(\{j^{(1)},\dots,j^{(k)} \})-u(\{1,\dots,j-1\})]\nonumber \\
		& = & \left(1-\frac{1}{k} \right)u(\{1,\dots,j-1 \}) + \frac{1}{k}a_{j}\nonumber \\
		& \ge & \left(1-\frac{1}{k} \right)u(\{1,\dots,j-1 \}) + \frac{1}{k}a_{1}\label{eq:2}
		\end{eqnarray}
where $(a)$ follows by submodularity of $u$ and $(b)$ follows by non-negativity and monotonicity of $u$. 
		
We now proceed with the cascading argument:
		\begin{align*}
		u(\{1,\dots,k\})  \geq & ~ \left(1-\frac{1}{k} \right)u(\{1,\dots,k-1 \}) + \frac{1}{k}a_{1}  \cr
		\geq & ~ \left(1-\frac{1}{k}\right)^2 u(\{1,\dots,k-2 \}) + \left(1-\frac{1}{k}\right)\frac{a_{1}}{k} + \frac{a_{1}}{k} \cr
		\geq & ~ \ldots \cr
		\geq & ~ \frac{a_{1}}{k}\left( \sum_{j=0}^{k-1} \left(1-\frac{1}{k}\right)^j \right) \cr
		\geq & ~ a_{1} \left( 1-\left(1-\frac{1}{k} \right)^{k} \right)\cr
		\geq & ~  \left(1-\frac{1}{e} \right)a_{1}.
		\end{align*}
		
		For the last step, we use the fact that $\left(1-1/k \right)^{k} \leq 1/e$, for all $k \geq 1$.
				
\paragraph{Proof of the upper bound.} Without loss of generality, assume that $S=\{1,2,\dots,k \}$ and $a_{1} \le a_{2} \le \dots \le a_{k}$. Recall that the value function $g$ is defined on $\reals^n$. We will slightly abuse notation by writing $g(\vec{y})$ to denote $g(\vec{y},\phi,\ldots,\phi)$, for any vector $\vec{y}$ of dimension $1 \leq d < n$, where $\phi$ is some minimal-value element defined in Section~\ref{sec:problem}. Moreover, for convenience, we will assume that the value function $g$ is continuous on any given dimension.

		
Define $g^{max}_i$ to be the maximum value of the submodular function $g$ on a vector of dimension $i$, i.e., 
$$
g^{\max}_i = \max_{z_1, z_2, \ldots, z_i\in \reals_+} g(z_1, z_2, \ldots, z_i).
$$

		
Suppose that $v = \min\{ c a_k, g^{\max}_{k-1}\}$, for some constant $c > 1$ whose value we will determine later. We first claim that there exists at least one vector $\vec{z}$ such that $g(\vec{z}) = v$. Our proof will leverage this vector $\vec{z}$ as follows. We consider a fictitious set of items $S^*$ whose individual performances correspond to $\vec{z}$ and show that the marginal benefit of adding an item $i \in N$ to this fictitious set is at most twice the marginal value of adding item $i$ to a set comprising of $k-1$ replicas of item $i$. This allows us to establish an upper bound in terms of the test scores. Although $g(\vec{z}) = v = ca_k$ is sufficient for our proof to hold, it is possible that the function $g$ is capped at a value smaller than $ca_k$ and there does not exist any $\vec{z}$ satisfying $g(\vec{z}) = ca_k$. To handle this corner case, we define $v$ to be the minimum of $ca_k$ and $g^{\max}_{k-1}$. 

We now prove the above claim that $v$ has a non-empty preimage under $g$. When $v = g^{\max}_{k-1}$, the claim follows trivially since by the definition of $g^{\max}_i$, there exists a $(k-1)$-dimension vector whose function value is $g^{\max}_{k-1}$. On the other hand, when $c a_k < g^{\max}_{k-1}$, this comes from continuity arguments since we know that there exist points in $\reals^{k-1}_+$ where $g$ evaluates to values greater than and smaller than $v$ respectively. In summary, there exists at least one point where the function evaluates to $v$. Since $g$ satisfies the extended diminishing returns property, we can abuse notation and infer from the definition that there exists a vector\footnote{Note that some elements of this vector can be $\phi$ or zero} $z \in \reals^{n-1}_+$ such that $g(\vec{z}) = v$ and for any $\vec{y} \in \reals^{k-1}_+$ having $g(\vec{y}) \leq g(\vec{z})$, it must be the case that
	\begin{equation}
	   g(\vec{z},x) - g(\vec{z}) \leq g(\vec{y},x) - g(\vec{y}), \hbox{ for all } x\in \reals_+.
\label{eqn_upperbound_extendedsub}
\end{equation}	

It is worth pointing out that while Definition~\ref{def:esa} guarantees that~\eqref{eqn_upperbound_extendedsub} holds when the vector $y$ is of dimension $n-1$, one can simply start with a $(k-1)$-dimension vector $\vec{y}$ and simply pad a sufficient number of $\phi$ elements to arrive upon a $(n-1)$-dimension vector whose value is still $g(\vec{y})$. Therefore, let $\vec{z} = (z_1, z_2, \ldots, z_{n-1})^\top$ be an arbitrary vector such that $g(\vec{z}) = v$ and that it satisfies~\eqref{eqn_upperbound_extendedsub} for any $\vec{y} \in \reals^{k-1}_+$, $x \geq 0$ as long as $g(\vec{y}) \leq g(\vec{z})$. Let $S^* = \{q_1, q_2, \ldots, q_{n-1}\}$ be a set of (fictitious) items such that $X_{q_j} = z_j$ with probability $1$ (performance of each of these fictitious items is deterministic). Therefore, the performance of the set of items $S^*$ is given by $$
		u(S^*) = g(\vec{z}) = \min\{c a_k, g^{max}_{k-1}\}.
$$ 
		
	 Since $u$ is a non-negative, increasing and submodular function, we have 
		\begin{eqnarray}
		u(S)& \leq & u(S^* \cup S) \\
		& \le & u(S^*) + \sum_{i=1}^{k} \left( u(S^* \cup\{i\}) - u(S^*) \right) \\
		&\leq & c a_{k} + \sum_{i=1}^{k} \left( u(S^* \cup\{i\}) - u(S^*) \right).
		\label{eq:5}
		\end{eqnarray}
		Let $X_i^{(1)}, X_i^{(2)}, \ldots, X_i^{(k)}$ be independent random variables with distribution $P_i$. Let $X_i = X_i^{(k)}$ and $\vec{Y}_i = (X_i^{(1)}, X_i^{(2)}, \ldots, X_i^{(k-1)})^\top$. Note that
		\begin{align}
		u(S^* \cup\{i\}) - u(S^*) & = \E\left[g(\vec{z},X_i) - g(\vec{z})\right]  \\
		& \stackrel{(a)}{=}~ \E\left[g(\vec{z},X_i) - g(\vec{z}) ~|~ g(\vec{Y}_i) \leq g(\vec{z}) \right] \\
		& \stackrel{(b)}{\leq}~ \E\left[g(\vec{Y}_i,X_i) -  g(\vec{Y}_i) ~|~ g(\vec{Y}_i) \leq g(\vec{z}) \right] \\
		& \leq \frac{u\left(\{i^{(1)},\dots,i^{(k)} \}\right) - u\left( \{i^{(1)},\dots,i^{(k-1)} \}\right)}{\Pr\big[g(\vec{Y}_i) \leq g(\vec{z})\big]} \\
		& \stackrel{(c)}{\le}~ \frac{1}{\Pr\big[g(\vec{Y}_i) \leq g(\vec{z})\big]}\frac{a_i}{k} \\
		&\stackrel{(d)}{\le}~ \left(1-\frac{1}{c}\right)^{-1}\frac{a_{k}}{k},
		\label{eq:6}
		\end{align}
where $(a)$ comes from the fact that, by definition, $X_i$ and $\vec{Y}_i$ are independent; the inequality $(b)$ follows from the extended diminishing returns property outlined in~\eqref{eqn_upperbound_extendedsub} for $\vec{y} = \vec{Y}_i$--note that for any instantiation $\vec{Y}_i$ where $g(\vec{Y}_i) \leq g(\vec{z})$, extended diminishing returns tells us that  $g(\vec{z},X_i) - g(\vec{z}) \leq g(\vec{Y}_i,X_i) -  g(\vec{Y}_i)$ for all $X_i$, thus taking the expectation over all $\vec{Y}_i, X_i$ conditional upon $g(\vec{Y}_i) \leq g(\vec{z})$ gives us $(b)$; inequality $(c)$ can be shown using only the definition of submodularity as can be seen via the below sequence of inequalities:
				\begin{align*}
				u\left(\{i^{(1)},\dots,i^{(k)} \}\right) - u\left( \{i^{(1)},\dots,i^{(k-1)}\right) & \leq \frac{1}{k}\sum_{j=0}^{k-1}\left(u(\{i^{(1)},\dots,i^{(j)},i^{(k)}\})   - u(\{i^{(1)},\dots,i^{(j)}\}) \right) \\
				& = \frac{1}{k}\sum_{j=0}^{k-1}\left(u(\{i^{(1)},\dots,i^{(j)},i^{(j+1)}\})   - u(\{i^{(1)},\dots,i^{(j)}\}) \right)\\
				& = \frac{1}{k}u(\{i^{(1)},\dots,i^{(k)}\})\\
				& = \frac{a_i}{k}.
				\end{align*}
							
It remains to prove (d) which follows by the fact $a_i\geq a_k$ for all $i\in \{1,2\ldots,k\}$ and showing that $\Pr\big[g(\vec{Y}_i) \leq g(\vec{z})\big] \geq 1-1/c$. Recall that $g(\vec{z}) = \min\{c a_k, g^{max}_{k-1}\}$. Let us proceed by separately considering two cases depending on the value of $g(\vec{z})$. If $g(\vec{z}) = g^{max}_{k-1}$, then $\Pr\big[g(\vec{Y}_i) \leq g(\vec{z})\big] =1$ trivially.  This is because by definition $g^{max}_{k-1}$ is the maximum value that the function can take for any vector of length $k-1$. On the other hand, when $g(\vec{z}) = c a_k$, we can apply Markov's inequality to obtain 					
\begin{equation*} 
\Pr[g(\vec{Y}_i)\geq c a_{k}] \leq \frac{\E [g(\vec{Y}_i)]}{c a_{k}}\leq \frac{\E [g(\vec{Y}_i,X_i)]}{c a_{k}}\le \frac{1}{c}. 
\end{equation*} 

Hence, it follows $\Pr[g(\vec{Y}_i)\leq c a_{k}]\geq 1 - \Pr[g(\vec{Y}_i)\geq c a_{k}]\geq 1-1/c$. Combining this with \eqref{eq:5} and \eqref{eq:6}, we obtain $u(S) \leq c a_k + (1-1/c)^{-1}a_k = (c^2/(c-1))a_k$. Since we can choose $c$ arbitrarily, by taking $c=2$, we obtain $u(S)\leq 4a_k$, which proves the upper bound.

\section{Proof of Lemma~\ref{clm_strskeproject}}
\label{app:strsketch}
		Suppose that $S^*$ is the optimum solution to the submodular welfare maximization problem with sketch utility functions $v_1, v_2,\ldots, v_m$, and  $v(S) \geq \alpha v(S^*)$. Then,				
		\begin{align*}
		u(\vec{\OPT}) &= \sum_{j=1}^m u_j(\OPT_j) \\
		& \leq \chigh \sum_{j=1}^m v_j(\OPT_j) \\
		& \leq \chigh \sum_{j=1}^m v_j(S^*_j) \quad \text{(since this solution is optimal for sketch utility functions)}\\
		& \leq \frac{1}{\alpha} \chigh \sum_{j=1}^m v_j(S_j)\\
		& \leq \frac{1}{\clow\alpha}\chigh \sum_{j=1}^m u_j(S_j).
		\end{align*}

\section{Proof of Lemma~\ref{thm:SMBsketch}}
\label{app_proof_lemmasmbsketch}

It suffices to consider an arbitrary partition $j$. To simplify the presentation, with a slight abuse of notation, we omit the index $j$ in our notation. 

Let $a_{1}^r, a_{2}^r,\ldots, a_{n}^r$ denote replication test scores for parameter $r$. For any set $S\subseteq N$ such that $|S| = k$, let $\pi(S) = (\pi_1(S), \pi_2(S),\ldots, \pi_k(S))$ be a permutation of the elements of $S$ defined in (\ref{equ:pi}). 

Let $v$ be a set function, which for any $S\subseteq N$ such that $|S| = k$ is defined by
\begin{equation}
v(S) = \sum_{r=1}^k \frac{1}{r}a_{\pi_r(S)}^r.
\label{equ:vstar}
\end{equation}

We need to establish the following relations, for every $S \subseteq N$,
\begin{equation}
\label{lowerbound}
u(S) \geq \frac{1}{2(\log(k)+1)} v(S)
\end{equation}
and
\begin{equation}
\label{upperbound}
u(S) \leq 6 v(S).
\end{equation}

\paragraph{Proof of lower bound (\ref{lowerbound})} Suppose that $S$ is of cardinality $k$ and define 
$$
\tau:= \argmax_t a_{\pi_t(S)}^t.
$$

We begin by noting the following basic property  of replication test scores. 

\begin{lemma} For replication test scores $a_1^r, a_2^r, \ldots, a_n^r$ for $1\leq r \leq k$, for every item $i\in \{1,2,\dots, k\}$, the following relations hold:
$$
\frac{a_{i}^s}{s} \geq \frac{a_{i}^t}{t}, \hbox{ for } 1 \leq s \leq t \leq k. 
$$
\label{lem_scoresubmodularity}
\end{lemma}

The assertion in Lemma~\ref{lem_scoresubmodularity} follows easily by the diminishing increments property of replication test scores $a_i^r$ with respect to parameter $r$.


In our proof, we will also need the following lemma:

\begin{lemma} For every set $S\subseteq N$ such that $|S| = k$ and ordering of items of this set $\pi(S) = (\pi_1(S), \pi_2(S), \ldots, \pi_k(S))$, the following relation holds:
$$
\frac{1}{\tau} \sum_{r=1}^{\tau} a_{\pi_{r}(S)}^r \geq \frac{1}{2}a_{\pi_{\tau}(S)}^\tau.
$$
\label{lem_scoremaxbound} 
\end{lemma}

The proof of the lemma is as follows. For every $r\in \{1,2,\dots,\tau\}$, we have
$$
\frac{a^r_{\pi_r(S)}}{r}\geq \frac{a^r_{\pi_\tau(S)}}{r}\geq \frac{a^\tau_{\pi_\tau(S)}}{\tau}
$$
where the first inequality is by definition of $\pi(S)$ and the second inequality is by Lemma~\ref{lem_scoresubmodularity}. Hence, we have

$$
\sum_{r=1}^{\tau} a_{\pi_r(S)}^r \geq \frac{a_{\pi_{\tau}(S)}^\tau}{\tau}\sum_{r=1}^{\tau} r \geq \frac{a_{\pi_{\tau}(S)}^\tau}{\tau}\frac{\tau(\tau+1)}{2} \geq a_{\pi_{\tau}(S)}^\tau \frac{\tau}{2}
$$
which corresponds to the claim of the lemma.

\begin{lemma} For every $S\subseteq N$, the following relation holds:
\begin{equation*}
u(S) \geq \frac{1}{\tau}\sum_{r=1}^{\tau} a_{\pi_{r}(S)}^r. 
\label{equ:uslb}
\end{equation*}
\label{lem:uslb}
\end{lemma}
The proof of Lemma~\ref{lem:uslb} is by induction as we show next. The inductive statement is $u(\{\pi_1(S), \ldots, \pi_r(S)\}) \geq \frac{1}{r}\sum_{s=1}^r a_{\pi_s(S)}^s$ for every $r\in \{1,2,\ldots, \tau\}$. Base case: $r = 1$. The base case indeed holds because by definition of replication test scores $u(\{\pi_1(S)\}) = a_{\pi_1(S)}^1$. Inductive step: assume that the statement is true up to $r-1$ and we need to show that it holds for $r$. We have the following relations: 
\begin{align*}
& u(\{\pi_1(S), \ldots, \pi_r(S)\}) - u(\{\pi_1(S), \ldots, \pi_{r-1}(S)\}) \\
& = \frac{1}{r}\Big(u(\{\pi_1(S), \ldots, \pi_{r-1}(S), \pi_r(S)^{(1)}\}) + \cdots + u(\{\pi_1(S), \ldots, \pi_{r-1}(S), \pi_r(S)^{(r)}\}) - r u(\{\pi_1(S), \ldots, \pi_{r-1}(S)\}) \Big)\\
& \geq \frac{1}{r}\Big(u(\{\pi_1(S), \ldots, \pi_{r-1}(S), \pi_r(S)^{(1)}, \ldots, \pi_r(S)^{(r)}\}) -  u(\{\pi_1(S), \ldots, \pi_{r-1}(S)\}) \Big) \\
& \geq  \frac{1}{r}\Big(u(\{\pi_r(S)^{(1)}, \ldots, \pi_r(S)^{(r)}\}) - u(\{\pi_1(S), \ldots, \pi_{r-1}(S)\}) \Big)\\
& = \frac{a_{\pi_r(S)}^r}{r} - \frac{u(\{\pi_1(S), \ldots, \pi_{r-1}(S)\})}{r}
\end{align*}
where the first and second inequality is by submodularity and monotonicity of set function $u$, respectively.

From the inductive hypothesis, we know that $u(\{\pi_1(S), \ldots, \pi_{r-1}(S)\}) \geq \frac{1}{r-1}\sum_{s=1}^{r-1}a_{\pi_s(S)}^s$, so we add $u(\{\pi_1(S), \ldots, \pi_{r-1}(S)\})$ to both sides of the above equation and obtain
$$
u(\{\pi_1(S), \ldots, \pi_{r}(S)\}) \geq \frac{a_{\pi_r(S)}^r}{r} + \frac{r-1}{r}u(\{\pi_1(S), \ldots, \pi_{r-1}(S)\}) \geq \frac{1}{r}\sum_{s=1}^r a_{\pi_s(S)}^s
$$
which proves the claim of Lemma~\ref{lem:uslb}. 

Now, combining Lemma~\ref{lem_scoremaxbound} and Lemma~\ref{lem:uslb}, we obtain $u(S) \geq a_{\pi_{\tau}(S)}^\tau/2$. 

Finally, we conclude the lower bound as follows:
$$
u(S) \geq \frac{1}{2}a_{\pi_{\tau}(S)}^\tau = \frac{a_{\pi_{\tau}(S)}^\tau}{2}  \frac{1 + \frac{1}{2} + \ldots + \frac{1}{k}}{1 + \frac{1}{2} + \ldots + \frac{1}{k}} \geq \frac{a_{\pi_1(S)}^1 + \frac{a_{\pi_2(S)}^2}{2} + \ldots + \frac{a_{\pi_k(S)}^k}{k}}{2(\log(k)+1)} = v(S)
$$
where in the last inequality we use the facts that $a_{\pi_{\tau}(S)}^\tau \geq a_{\pi_{r}(S)}^r$ for all $r$, and $1+1/2+\cdots + 1/k \leq \log(k)+1$, for all $k\geq 1$.

\paragraph{Proof of the upper bound (\ref{upperbound})} The proof of the upper bound is almost identical to the upper bound proof of Lemma~\ref{thm:subtest}. Once again, we will abuse notation by writing $g(\vec{y})$ instead of $g(\vec{y},\phi,\ldots,\phi)$ for any vector $\vec{y}$ of dimension $r < n$, where $\phi$ is some minimal-value element as defined in Section~\ref{sec:problem}.

Analogous (but slightly different) than in the proof of Lemma~\ref{thm:subtest}, consider a deterministic vector $\vec{z} = (z_1, z_2, \ldots, z_{n-1})$ such that $g(\vec{z}) = \min\{c a_{\pi_{\tau}(S)}^\tau, g^{max}_{k-1}\}$, for a positive constant $c > 1$ whose value will be determined later. In choosing this vector, we will apply the definition of extended diminishing returns so that for any $\vec{y}$ satisfying $g(\vec{y}) \leq g(\vec{z})$ and $x \geq 0$, Equation~\eqref{eqn_upperbound_extendedsub} is satisfied. 

Let $S^* = \{v_1, v_2, \ldots, v_{n-1}\}$ be a set of (fictitious) items such that $X_{v_j} = z_j$ with probability $1$ (the performance of each of these fictitious items is deterministic). Therefore, the performance of the set of items $S^*$ is given by $u(S^*) = g(\vec{z}) = \min\{c a_{\pi_{\tau}(S)}^\tau, g^{max}_{k-1}\}$. 

By definition, we know that $a_{\pi_{r}(S)}^r \leq a_{\pi_{\tau}(S)}^\tau$ for all $r$. Moreover, we can upper bound $u(S)$ as follows,
\begin{equation}
u(S) \leq u(S \cup S^*) \leq u(S^*) + \sum_{r=1}^{k}[u(S^* \cup \{\pi_r(S)\}) - u(S^*)].
\label{equ:uulb}
\end{equation}
Let $X_{\pi_r(S)}^{(1)}, X_{\pi_r(S)}^{(2)}, \ldots, X_{\pi_r(S)}^{(r)}$ be independent random variables with distribution $P_{\pi_r(S)}$. Let $X = X_{\pi_r(S)}^{(r)}$ and $\vec{Y} = (X_{\pi_r(S)}^{(1)}, X_{\pi_r(S)}^{(2)},\ldots,X_{\pi_r(S)}^{(r-1)})$. Note that
\begin{align*}
u(S^* \cup \{\pi_r(S)\}) - u(S^*) & = \E [g(\vec{z},X) - g(\vec{z})]\\
	& \stackrel{}{=}~ \E\left[g(\vec{z},X) - g(\vec{z}) ~|~ g(\vec{Y}) \leq g(\vec{z}) \right] \\
& \stackrel{(a)}{\leq} \E \big[g(\vec{Y},X) -  g(\vec{Y}) ~|~ g(\vec{Y}) \leq  g(\vec{z})\big]\\
& \leq \frac{\E [g(\vec{Y},X) -  g(\vec{Y})]}{\Pr[g(\vec{Y}) \leq  g(\vec{z}) ]} \\
& \stackrel{(b)}{\leq} \frac{1}{\Pr [g(\vec{Y}) \leq g(\vec{z}) ]}\frac{a_{\pi_r(S)}^r}{r}. \label{eqn_markovpre}
\end{align*}

Inequality $(a)$ follows from the extended diminishing returns property defined in Definition~\ref{def:esa}. Note that from our definition of $\vec{z}$, for any instantiation $\vec{Y}$ where $g(\vec{Y}) \leq g(\vec{z})$, extended diminishing returns tells us that  $g(\vec{z},X) - g(\vec{z}) \leq g(\vec{Y},X) -  g(\vec{Y})$ for all $X$. Taking the expectation over all $\vec{Y}, X$ conditional upon $g(\vec{Y}) \leq g(\vec{z})$ gives us $(a)$.

Inequality $(b)$ can be shown using only the definition of submodularity as can be seen via the below sequence of inequalities: suppose that $i = \pi_r(S)$.
\begin{align*}
\E [g(\vec{Y},X) -  g(\vec{Y})] & \leq \frac{1}{r}\sum_{s=0}^{r-1}\left(u(\{i^{(1)},\dots,i^{(s)},i^{(r)}\})   - u(\{i^{(1)},\dots,i^{(s)}\}) \right) \\
& = \frac{1}{r}\sum_{s=0}^{r-1}\left(u(\{i^{(1)},\dots,i^{(s)},i^{(s+1)}\})   - u(\{i^{(1)},\dots,i^{(s)}\}) \right)\\
& = \frac{1}{r}u(\{i^{(1)},\dots,i^{(r)}\})\\
& = \frac{a^{r}_{i}}{r} =  \frac{a^{r}_{\pi_r(S)}}{r}.
\end{align*}
All that remains for us is to prove that $\Pr\big[g(\vec{Y}) \leq g(\vec{z})\big] \geq 1-1/c$.

Recall that $g(\vec{z}) = \min\{c a^{\tau}_{\pi_{\tau}(S)}, g^{max}_{k-1}\}$. Let us proceed by considering two cases depending on the value of $g(\vec{z})$. If $g(\vec{z}) = g^{max}_{k-1}$, then $\Pr\big[g(\vec{Y}) \leq g(\vec{z})\big] =1$ trivially.  This is because by definition $g^{max}_{k-1}$ is the maximum value that the function can take on any vector of length $k-1$, and by monotonicity, any vector of size $r-1$ such as $\vec{Y}$ since $r \leq k$. On the other hand, when $g(\vec{z}) = c a_{\pi_{\tau}(S)}^\tau$, we can apply Markov's inequality and bound the desired probability, i.e., 					
$$
\Pr \left[g(\vec{Y}) \geq  c a_{\pi_{\tau}(S)}^\tau \right] \leq \frac{\E \left[g(\vec{Y})\right]}{c a_{\pi_{\tau}(S)}^\tau} \leq \frac{1}{c}
$$
where we used $\E[g(\vec{Y})] = a_{\pi_r(S)}^{r-1} \leq a_{\pi_r(S)}^r \leq a_{\pi_{\tau}(S)}^\tau$. Since $\Pr \left[g(\vec{Y}) \leq  c a_{\pi_{\tau}(S)}^\tau \right] \geq 1 - \Pr \left[g(\vec{Y}) \geq  c a_{\pi_{\tau}(S)}^\tau \right]$, it follows that $\Pr \left[g(\vec{Y}) \leq  c a_{\pi_{\tau}(S)}^\tau \right]\geq 1-1/c$, as desired. 

We have shown that $u(S^* \cup \{\pi_r(S)\}) - u(S^*) \leq (1-1/c)^{-1} a_{\pi_r(S)}^r/r$.

Combining with (\ref{equ:uulb}), we obtain
$$
u(S) \leq c a_{\pi_{\tau}(S)}^\tau + \left(1-\frac{1}{c}\right)^{-1}\left(\frac{a_{\pi_1(S)}^1}{1} + \frac{a_{\pi_2(S)}^2}{2}  + \cdots + \frac{a_{\pi_k(S)}^k}{k}\right).
$$ 

Applying Lemma~\ref{lem_scoremaxbound} to $a_{\pi_{\tau}(S)}^\tau$, we obtain that 

\begin{align*}
u(S) & \leq 2c\frac{1}{\tau}\sum_{r=1}^{\tau} a_{\pi_{r}(S)}^r + \left(1-\frac{1}{c}\right)^{-1}\left(a_{\pi_1(S)}^1 + \frac{a_{\pi_2(S)}^2}{2} + \cdots + \frac{a_{\pi_k(S)}^k}{k} \right)\\
& \leq \left(2c + \left(1-\frac{1}{c}\right)^{-1}\right)\left( a_{\pi_1(S)}^1 + \frac{a_{\pi_2(S)}^2}{2} + \cdots + \frac{a_{\pi_k(S)}^k}{k} \right)
\end{align*}
which completes the proof by taking $c = 2$.

\section{Proof of Lemma~\ref{lem:2}}
\label{app:2}

Before proving Lemma~\ref{lem:2}, we prove that our sketch function $v_j$ as defined in~\eqref{vs} satisfies a simple monotonicity property. This property will be useful in the proof of Lemma~\ref{lem:2}.
\begin{proposition}
\label{prop_lem7_monotonicity}
Suppose $v_j$ is a sketch function for a stochastic monotone submodular function $u_j$ as defined in~\eqref{vs} and let $S = \{i_1,i_2,\ldots, i_{|S|}\} \subseteq N$ such that for all $r\in \{1,2,\ldots, |S|\}$, $\pi_r(S,j) = i_r$. Then, the following inequalities hold for all $r\in \{1,2,\ldots, |S|\}$:
$$v_j(S) \geq v_j(S \setminus \{i_r\}) \geq v_j(S) - \frac{a^r_{i_r,j}}{r}.$$
\end{proposition}

\proof{Proof of Proposition~\ref{prop_lem7_monotonicity}}
Fix some $r\in \{1,2,\ldots, |S|\}$, and for all $t \neq r$, define $\nu_t$ such that $\pi_{\nu_t}(S\setminus \{i_r\},j) = i_t$. That is, $\nu_t$ denotes item $i_t$'s new `rank' in the set $S\setminus \{i_r\}$. Note that $1 \leq \nu_t \leq |S|-1$ and that:
\begin{equation}
\label{eqn_prop_newvalue}
 v_j(S\setminus \{i_r\}) = \sum_{t \neq r}\frac{a^{\nu_t}_{i_t,j}}{\nu_t}.
 \end{equation}

We show via induction on $t$ that for all $t \neq r$, $\nu_t \leq t$, i.e., removal of an item cannot hurt the `rank' of another item. The claim is trivially true when $t=1$ since $\nu_t \geq 1$. Consider an arbitrary $t > 1$, and suppose that the inductive hypothesis is true up to $t-1$. Let us consider two cases: first, if $t < r$, then by definition $\pi_t(S,j) = \pi_{t}(S\setminus \{i_r\}, j) = i_t$ and so the inductive claim holds since $\nu_t = t$. Second, suppose that $t > r$: assume by contradiction that $\nu_t > t$. By the inductive hypothesis, it must be the case that $\pi_{t}(S\setminus \{i_r\},j) \in \{i_t, i_{t+1}, \ldots, i_{|S|}\}$---indeed, for all $t' < t$, we have that $\nu_{t'} \leq t'$. However, we know by definition of $\pi$ that for all $i \in \{i_{t+1}, \ldots, i_{|S|}\}$, it must be true that:
$$a^t_{i_t,j} > a^t_{i,j}.$$

Therefore, if $\nu_t > t$, then $\pi_{t}(S\setminus \{i_r\},j) \in \{i_{t+1}, \ldots, i_{|S|}\}$---this would be a violation of the definition of $\pi$. Hence, the inductive hypothesis follows. 

Now, in order to prove the proposition, we go back to~\eqref{eqn_prop_newvalue},
\begin{align*}
v_j(S\setminus \{i_r\}) & = \sum_{t \neq r}\frac{a^{\nu_t}_{i_t,j}}{\nu_t} \\
& \leq \sum_{t \neq r}\frac{a^{\nu_t}_{i_{\nu_t},j}}{\nu_t} \\
& = \sum_{t=1}^{|S|-1}\frac{a^{t}_{i_{t},j}}{t} \\
& \leq v(S).
\end{align*}
The crucial step above is the second inequality. There, we used the fact that $\nu_t \leq t$, and therefore, if $\nu_t = q$, then $a^{q}_{i_{q}, j} \geq a^{q}_{i_t, j}$ by definition of $i_q$ for all $1 \leq q = \nu_t \leq |S|-1$. The third inequality comes from changing the index from $\nu_t$ to $t$. In summary, we have shown that $v(S) \geq v(S\setminus \{i_r\})$ which is one half the proposition. In order to prove the other half, that is  $v(S\setminus \{i_r\}) \geq v(S) - a^r_{i_r,j}/r$, we utilize the result from Lemma~\ref{lem_scoresubmodularity}, namely that:
$$\frac{a^{\nu_t}_{i_t,j}}{\nu_t} \geq \frac{a^{t}_{i_t,j}}{t},$$
which is true because $\nu_t \leq t$. To conclude the proposition, we have that:
\begin{align*}
v(S) & = \sum_{t=1}^{|S|}\frac{a^{t}_{i_{t},j}}{t} \\
& = \sum_{t\neq r}\frac{a^{t}_{i_{t},j}}{t} + \frac{a^r_{i_r,j}}{r} \\
& \leq \sum_{t\neq r}\frac{a^{\nu_t}_{i_{t},j}}{\nu_t} + \frac{a^r_{i_r,j}}{r} \\
& = v(S\setminus \{i_r\}) + \frac{a^r_{i_r,j}}{r}.\Halmos
\end{align*}
\endproof

We are now ready to prove the main lemma.
\proof{(Proof of Lemma~\ref{lem:2})}We need to show that the greedy algorithm described in Algorithm~\ref{alg_greedyassignment} returns an assignment $\vec{S}=(S_1, S_2, \ldots, S_m)$ that is a $\frac{1}{2}$-approximation to the optimum assignment $\vec{O}=(O_1, O_2, \ldots, O_m)$ that maximizes $v(\vec{S}') = \sum_{j=1}^m v_j(S'_j)$ where the function $v_j$ is as defined in~\eqref{vs}. If the sketch function $v_j$ is submodular, then one can simply apply the well-known result by~\cite{lehmann2006combinatorial} for the submodular welfare maximization problem to show that the greedy algorithm yields the desired approximation factor. However, despite its simplicity, the sketch function $v_j$ is not necessarily submodular, so we cannot directly use the existing proof for submodular welfare maximization as a black-box.

Before proving the result, we introduce some pertinent notation. Recall that our algorithm proceeds in rounds such that at each time step $t$, exactly one item $i \in A$ is added to a partition $j \in P$. Let $\vec{S}(t) = (S_1(t), S_2(t), \ldots, S_m(t))$ denote the assignment at the end of time step $t$, i.e., $S_j(t)$ is the set of items assigned to partition $j \in M$ at the end of $t$ unique assignments. For notational convenience, let $\vec{S}(0) = (\emptyset, \emptyset,\ldots, \emptyset)$. Suppose that $\vec{O}(t) = (O_1(t), O_2(t), \ldots, O_m(t))$ denote the optimal (constrained) assignment such that for every $j \in M$, $S_j(t) \subseteq O_j(t)$, i.e., this assignment deviates from $\vec{S}$ only in the set of items that are unassigned at the end of time step $t$. Finally, suppose that at round $t+1$, if our algorithm assigns item $i \in N$ to partition $j \in M$, then the added welfare is $\Delta(t+1) := a^{|S_j(t)|+1}_{i,j}/(|S_j(t)|+1)$.

The basic idea behind our proof is similar to that of Theorem 12 in~\citep{lehmann2006combinatorial}. Namely, we show that $v(\vec{O}(t)) \leq v(\vec{O}(t+1)) + \Delta(t+1)$ for all $t \in \{0,1,\ldots,\ell-1\}$, where $\ell$ is the total number of rounds the algorithm proceeds for. By cascading this argument, we can show the desired approximation guarantee, i.e.,
\begin{align}
v(\vec{O}(0)) & \leq v(\vec{O}(1)) + \Delta(1) \\
& \leq \cdots \nonumber \\
& \leq v(\vec{O}(t)) + \sum_{r=1}^t \Delta(r) \nonumber\\
& \leq \cdots \nonumber \\
& \leq v(\vec{O}(\ell)) + \sum_{r=1}^\ell \Delta(r) \nonumber \\
& = v(\vec{O}(\ell)) + v(\vec{S}(\ell)) \nonumber \\
& = 2v(\vec{S}).
\end{align}
The first five equations above come from an application of the claimed inequality $v(\vec{O}(t)) \leq v(\vec{O}(t+1)) + \Delta(t+1)$ for all $t \in \{0,1,\ldots,\ell-1\}$. The penultimate and final equations follow from: (a) $\vec{O}(\ell) = \vec{S}(\ell) = \vec{S}$ by definition, and (b) the total welfare generated by the solution $\vec{S}$ is simply the sum of welfare added in each round, i.e., $\sum_{r=1}^{\ell} \Delta(r)$. Finally, this argument can be used to conclude the proof since $\vec{O}(0)$ is the same as the unconstrained optimum assignment $\vec{O}$ by definition.

All that remains for us is to prove the claim  $v(\vec{O}(t)) \leq v(\vec{O}(t+1)) + \Delta(t+1)$ for all $t \in \{0,1,\ldots,\ell-1\}$. In ~\citep{lehmann2006combinatorial}, this claim followed from submodularity. However, since this is no longer a valid approach in our setting, we use a more subtle argument based on the monotonicity result from Proposition~\ref{prop_lem7_monotonicity}. 

Suppose that in round $t+1$, our algorithm assigns item $i$ to partition $j$ and let $|S_j(t+1)| = r$ so that $\Delta(t+1) = a^r_{i,j}/r$. Moreover, suppose that in the constrained optimum solution $\vec{O}(t)$, item $i$ is assigned to partition $j'$ and integer parameter $r'$ is such that $\pi_{r'}(O_{j'}(t),j') = i$. A crucial observation here is that $r' > |S_{j'}(t)|$. Indeed, since $S_{j'}(t) \subseteq O_{j'}(t)$, if $r' \leq |S_{j'}(t)|$, then it would be the case that\footnote{For convenience, we assume no ties here although the proof can easily be extended to the case with ties as long as a consistent tie-breaking rule is used.} $$a^{r'}_{i,j'} > a^{r'}_{\pi_{r'}(S_{j'}(t), j'),j'}.$$

This is naturally a contradiction since Algorithm~\ref{alg_greedyassignment} greedily assigns the item with the maximum marginal benefit at each round and we know that item $i$ was still unassigned at the end of round $t$. Consider the assignment $\vec{O}(t+1)$, we have that:
\begin{equation}
\label{eqn_lem2_optt1}
v(\vec{O}(t+1)) \geq v(\vec{O}(t)) + \big(v_j(O_j(t) \cup \{i\})  - v_j(O_j(t))\big) - \big(v_{j'}(O_{j'}(t))  - v_{j'}(O_{j'}(t) \setminus \{i\})\big).
\end{equation}
Starting with the assignment $\vec{O}(t)$, if we move item $i$ from partition $j'$ to partition $j$, the resulting assignment has a welfare that is denoted by the right hand side of the above inequality. Now, since the resulting assignment also subsumes $\vec{S}(t+1)$, its welfare cannot be larger than $\vec{O}(j+1)$. Consider the term,  $v_j(O_j(t) \cup \{i\})  - v_j(O_j(t))$ from the RHS of~\eqref{eqn_lem2_optt1}---this is non-negative by the monotonicity argument laid out in Proposition~\ref{prop_lem7_monotonicity}. Similarly, consider the other term from the RHS, namely $v_{j'}(O_{j'}(t))  - v_{j'}(O_{j'}(t) \setminus \{i\})$---this is upper bounded by $a^{r'}_{i,j'}/r'$ as per Proposition~\ref{prop_lem7_monotonicity} and our definition of $r'$. Further, according to Lemma~\ref{lem_scoresubmodularity}, we have that:
$$\frac{a^{r'}_{i,j'}}{r'} \leq \frac{a^{|S_{j'}(t)|+1}_{i,j'}}{|S_{j'}(t)|+1},$$

since we proved earlier that $r' > S_{j'}(t)$. Putting all these ingredients together, we arrive upon the desired claim that $v(\vec{O}(t)) \leq v(\vec{O}(t+1)) + \Delta(t+1)$ for all $t \in \{0,1,\ldots,\ell-1\}$:
\begin{align}
v(\vec{O}(t+1)) & \geq v(\vec{O}(t)) + \big(v_j(O_j(t) \cup \{i\})  - v_j(O_j(t))\big) - \big(v_{j'}(O_{j'}(t)) - v_{j'}(O_{j'}(t) \setminus \{i\})\big)  \nonumber \\
& \geq v(\vec{O}(t)) +(0) - \frac{a^{r'}_{i,j'}}{r'} \label{eqn_finallem71}\\
& \geq v(\vec{O}(t))  - \frac{a^{|S_{j'}(t)|+1}_{i,j'}}{|S_{j'}(t)|+1} \label{eqn_finallem72}\\
& \geq v(\vec{O}(t))  - \frac{a^r_{i,j}}{r} \label{eqn_finallem73}\\
& = v(\vec{O}(t)) - \Delta(t+1) \nonumber.
\end{align}

Equation~\eqref{eqn_finallem71} is a product of the monotonicity claims from Proposition~\ref{prop_lem7_monotonicity}. Equation~\eqref{eqn_finallem72} is due to the fact that $r' > |S_{j'}(t)|$ and due to Lemma~\ref{lem_scoresubmodularity}. Finally, the penultimate inequality~\eqref{eqn_finallem73} comes from the property of the greedy algorithm. At round $t+1$, since the greedy algorithm assigned item $i$ to partition $j$ as opposed to partition $j'$, it must have been the case that $a^{|S_{j'}(t)|+1}_{i,j'}/(|S_{j'}(t)|+1) \leq  a^r_{i,j}/r$.  This concludes our proof.  \Halmos\endproof

\section{Sample Average Approximation Algorithms}		
\label{app:saa}

\subsection{NP-Hardness of Sample-Based Stochastic Optimization}
We now present an example of a stochastic submodular optimization problem with a rather simple utility function where employing sample based algorithms may subsequently result in a discrete optimization that is NP-Hard. On the other hand, test score algorithms avoid the additional overhead brought about by solving secondary optimization problems. More concretely, consider the problem of maximizing a stochastic monotone submodular function subject to a cardinality constraint where $g(\vec{x}) = \max\{x_1, x_2, \ldots, x_n\}$. For every $i \in N$, the distribution $P_i$ is defined as follows: let $X_i$ be a random variable such that $X_i = 1$ with probability $p_i$ and $X_i = 0$ with probability $1-p_i$ for some sufficiently small probabilities $(p_i)_{i \in N}$. 

Consider the sample average approximation approach which first computes a collection of $T$ independent sample vectors $(X^{(t)}_1, X^{(t)}_2, \ldots, X^{(t)}_n)_{t=1}^T$, where $X^{(t)}_i \sim P_i$. For a given cardinality parameter $k$, the SAA method would look to compute a subset $S^* \subseteq N$ in order to maximize the number of `covered indices' $t$, i.e., $\arg\max_{S \subseteq N}\sum_{t=1}^T \mathds{1}\{\exists i \in S: X^{(t)}_i = 1\}$, where $\mathds{1}$ is the indicator function that evaluates to one when the condition inside is true and is zero otherwise. However, this is equivalent to the well-studied maximum coverage problem which is known to be NP-Hard. Note that for the same instance, a test score algorithm based on replication scores would return the optimum solution with high probability since the test scores would be monotonically increasing in the probability $p_i$. In the following section, we delve deeper into the sample errors due to test score and SAA methods 

\subsection{Error Probability for Finite Samples}
We discuss the use of sample averages for estimating test scores for the simple example introduced in Example~\ref{example:contentreco} and the numerical results provided in Section~\ref{sec:disc1}. Our goal is to characterize the probability of error in identifying an optimal set of items due to use of sample averages for approximating replication test scores for the aforementioned simple example. The simplicity of this example allows us to derive tight characterizations of the required number of samples for the probability of error to be within a prescribed bound. We also conduct a similar analysis for the sample averaging approach (SAA) that amounts to enumerating and estimating value of each feasible set of items, and compare with the test score based approach. 

Recall that we consider a ground set of items $N$ that consists of type-A and type-B items that reside in two disjoint nonempty sets $A$ and $B$, respectively, such that $N = A \cup B$. For each $i\in A$, $X_i = a$ with probability $1$, and for each $i\in B$, $X_i = b/p$ with probability $p$, and $X_i = 0$ otherwise, where $a,b > 0$ and $p \in (0,1]$ are parameters. We assume that $b/p > a$ so that individual performance of a type-$B$ item is larger than that of any type-$A$ item conditional on the type-$B$ item achieving performance $b/p$. We may think of type-$B$ items as of high-risk, high-return items when $p$ is small. We assume that for given $k$, $|A| \geq k$ and $|B| \geq k$.

We consider the best-shot utility function $u(S) = \E[\max\{x_i\mid i\in S\}]$, which want to maximize over sets $S\in 2^N$ of cardinality $|S| = k$. Clearly, we can distinguish $k + 1$ equivalence cases for sets $S$ with respect to the value of the utility function: class $r$ defined by having $r$ type-$B$ items and $k-r$ type-$A$ items, for $r \in \{0, 1,\ldots, k\}$. Let $C_{k,r}$ denote all sets of cardinality $k$ that are of class $r$.

For each $S \in C_{k,r}$, we have
$$
u(S) = \E[\max\{X_{(r)},a\}]
$$
where $X^{(r)}$ is the largest order statistic of individual performance of type-$B$ items, 
$$
\Pr[X_{(r)}=b/p] = 1-\Pr[X_{(r)} = 0] = 1 -(1-p)^r.
$$
Indeed, we have
$$
u(S) = a (1-p)^r + \frac{b}{p}(1-(1-p)^r).
$$
Since we assumed that $b/p > a$, we have that $u(S)$ is increasing in the class of set $S$, achieving the largest value for $r = k$, i.e. when all items are of type $B$.

In our analysis, we will make use of the well-known Hoeffding's inequality \citep{H63} to bound the probability of the event that a sum of independent random variables with bounded supports deviates from its expected value by more than a given amount. 

\begin{proposition}[Hoeffding's inequality] Let $X_1, X_2, \ldots, X_T$ be independent random variables such that $X_i \in [\alpha_i,\beta_i]$ with probability $1$ for all $i\in \{1,2,\ldots,T\}$. Then, for every $x\geq 0$,
$$
\Pr[X_1 + X_2 +\cdots + X_T - \E[X_1 + X_2 + \cdots + X_T] \geq x] \leq \exp\left(-\frac{2x^2 T^2}{\sum_{i=1}^T (\beta_i - \alpha_i)^2}\right).
$$
\end{proposition}

\paragraph{Test scores} Consider sample average estimators of replication test scores defined as follows:
$$
\hat{a}_i = \frac{1}{T}\sum_{t=1}^T\max\{X_i^{(1,t)}, X_i^{(2,t)}, \ldots, X_i^{(k,t)}\}
$$
where $X_i^{(j,t)}$ are independent over $i$, $j$, and $t$ and $X_{i}^{(j,t)}$ has distribution $P_i$. Indeed, by denoting $X_i^{((k),t)}$ the largest order statistic of $P_i$, we can write
$$
\hat{a}_i = \frac{1}{T} \sum_{t=1}^T X_i^{((k),t)}.
$$
Indeed, for our example, for every $i\in A$, we have $\hat{a}_i = a$. On the other hand, for every $i\in B$, we have that $X_i^{((k),t)}$ is equal to $b/p$ with probability $1-(1-p)^k$ and is equal to $0$ otherwise. Thus, for every $i\in B$, $a_i = \E[\hat{a}_i] = (b/p)(1-(1-p)^k)$. In what follows, we assume that $(b/p)(1-(1-p)^k) > a$, i.e. $\E[\hat{a}_i] < \E[\hat{a}_j]$ for every $i\in A$ and $j\in B$. In this case, in absence of estimation noise, the replication test score based algorithm correctly identifies an optimum set of items to be a set $k$ type-$B$ items. We declare an error event to occur if $\hat{a}_j < \hat{a}_i$ for some items $i\in A$ and $j\in B$, and denote with $p_e$ the probability of this event, i.e. $p_e :=\Pr[\cup_{i\in A, j\in B}\{\hat{a}_j < \hat{a}_i\}]$.

By the Hoeffding's inequality, for any type-$A$ item $i$ and type-$B$ item $j$, we have
$$
\Pr[\hat{a}_j < \hat{a}_i] = \Pr\left[\frac{1}{T} \sum_{t=1}^T X_i^{((k),t)} < a\right] \leq \exp(-2(1-(1-p)^k -ap/b)^2T).
$$ 
By the union bound, we have
$$
p_e \leq |A||B|\exp\left(-2p^2\left((1-(1-p)^k)/p-a/b\right)^2T\right).
$$
Hence, for $p_e \leq \delta$ to hold, for given $\delta \in (0,1]$, it suffices that the total number of samples $m := nkT$ is such that
\begin{equation}
m \geq \frac{nk}{2 p^2 \left((1-(1-p)^k)/p-a/b\right)^2}\log\left(\frac{|A||B|}{\delta}\right).
\label{equ:m2}
\end{equation}

Under given assumptions $|A|+|B| = n$ and $|A|,|B|\geq k$, we have $|A| |B|\leq n^2/4$, so in (\ref{equ:m2}), we can replace $\log(|A||B|/\delta)$ with $2\log(n/2) + \log(1/\delta)$ to obtain a sufficient number of samples. Note that $((1-(1-p)^k)/p-a/b)^2 = (k-a/b)^2(1+o(1))$ for small $p$. Hence, we have $m = \Omega(1/p^2)$. 

\paragraph{SAA approach} Consider now a stochastic average approximation method that amounts to enumerating all feasible sets and then choosing the one that has the best estimated value: for each $S\subseteq N$ such that $|S|=k$, estimating $u(S)$ with the sample average $\hat{u}(S)$ defined as
$$
\hat{u}(S) = \frac{1}{T}\sum_{t=1}^T \max\{X_i^{(t)}\mid i\in S\}
$$
where $X_i^{(t)}$ are independent random variables over $i$ and $t$ and $X_{i}^{(t)} \sim P_i$ for all $S\in 2^N$ and $t \in \{1,2\ldots,T\}$.

For every class-$0$ set $S$, whose all elements are of type $A$, we have $\hat{u}(S) = a$ with probability $1$. For every class-$r$ set $S$, with $1\leq r < k$, we have $\hat{u}(S)\geq a$. For every class-$r$ set $S$, with $0\leq r \leq k$, we have
$$
\hat{u}(S) = a\left(1-\frac{X_S}{T}\right) + \frac{b}{p}\frac{X_S}{T} 
$$
where $X_S\sim \mathrm{Bin}(T,1-(1-p)^r)$. 

Comparing $\hat{u}(S) > \hat{u}(S')$ for any two sets $S$ and $S'$ is equivalent to $X_S > X_{S'}$. By the Hoeffding's inequality, for an two sets $S$ and $S'$ such that $\E[X_S] > \E[X_{S'}]$, we have
\begin{equation}
\Pr[X_S \leq X_{S'}] \leq \exp\left(-\frac{1}{2}(\E[X_S]-\E[X_{S'}])^2 T\right).
\label{equ:x1x2}
\end{equation}

We declare an error event to occur if $\hat{u}(S) < \hat{u}(S')$ for every class $k$ set $S$ and some class $r<k$ set $S'$ and denote with $p_e$ the probability of this event. Then, by the union bound, we have
\begin{eqnarray*}
p_e &=& \Pr[X_{S} < X_{S'} \hbox{ for every } S\in C_{k,k} \hbox{ and some } S'\in \cup_{0\leq r<k}C_{k,r}]\\
&\leq & \Pr[\cup_{S'\in \cup_{0\leq r<k}C_{k,r}} \{X_{S_k} < X_{S'}\}] \\
&\leq & \sum_{r=0}^{k-1} |C_{k,r}| \Pr[X_{S_k} \leq X_{S_r}]\\
&\leq & \left(\sum_{r=0}^{k-1} |C_{k,r}|\right) \Pr[X_{S_k} \leq X_{S_{k-1}}]\\
&=& \left(\binom{n}{k} - \binom{|B|}{k}\right) \Pr[X_{S_k} \leq X_{S_{k-1}}]
\end{eqnarray*}
where $S_i$ denotes an arbitrarily fixed set in $C_{k,i}$. 

Combining with (\ref{equ:x1x2}), we have
\begin{equation}
p_e \leq \left(\binom{n}{k} - \binom{|B|}{k}\right)\exp\left(-\frac{1}{2}p^2(1-p)^{2(k-1)}T\right).
\label{equ:pe1}
\end{equation}

Note that the error exponent in (\ref{equ:pe1}) is due to discriminating a class $k$ set from a class $k-1$ set. In order to have $p_e\leq \delta$, for given $\delta\in (0,1]$, it suffices for the total number of samples $m :=nT$ to be such that
\begin{equation}
m \geq \frac{2n}{p^2(1-p)^{2(k-1)}}\log\left(\frac{\binom{n}{k} - \binom{|B|}{k}}{\delta}\right).
\label{equ:m1}
\end{equation}

Note that in (\ref{equ:m1}) we can replace $\binom{n}{k} - \binom{|B|}{k}$ with $\binom{n}{k}$, which is tight for $|B| = \Theta(k)$. Furthermore, we can use the well known inequalities $k \left(\log\left(\frac{n}{k}\right)\right) \leq \log\left(\binom{n}{k}\right) \leq k \left(\log\left(\frac{n}{k}\right)+1\right)$. Thus, the logarithmic term in (\ref{equ:m1}) contributes a factor of $k$ to the sufficient number of samples. Note also that $m = \Omega(1/p^2)$.

\paragraph{Summary} The analysis of the estimation error for the SAA approach requires to consider discrimination of a set with all type-$B$ items and a set that has at least one type-$A$ item. On the other hand, for the approach based on using replication test scores, we only need to consider discrimination of a set with all type-$B$ items and a set with all type-$A$ items. For both approaches, we obtain that the error exponent scales as $\Theta(p^2)$ for small $p$. The SAA approach can require a larger number of samples than the replication test score approach, which is demonstrated by numerical results in Section~\ref{sec:disc1}. 

\section{Proof of Proposition~\ref{thm:cesmeantestscores}} 
\label{sec:cesmeantestscores}
	
	Let $X_1, X_2, \ldots,X_n$ be independent random variables with distributions $P_1, P_2, \ldots, P_n$, respectively, and let $\vec{X} := (X_1, X_2,\ldots,X_n)$. Without loss of generality, assume that items are enumerated in decreasing order of mean test scores, i.e. $\E[X_1]\geq \E[X_2]\geq \cdots \geq \E[X_n]$.  Let $S = \{i_1,i_2,\ldots,i_k\}$ be an arbitrary subset of items in $N$. Then, we have
			\begin{eqnarray}
			u(S) &=& \E[g(\perfmap_S(\vec{X}))]\nonumber\\
			&=& \E[[g(\perfmap_{S}(\vec{X})) - g(\perfmap_{S\setminus \{i_k\}}(\vec{X}))] + [g(\perfmap_{S\setminus \{i_k\}}(\vec{X})) - g(\perfmap_{S\setminus \{i_{k-1},i_k\}}(\vec{X}))] +\cdots+ [g(\perfmap_{\{i_1\}}(\vec{X}))-g(\phi,\ldots,\phi)]]\nonumber\\
&=& [u(S)-u(S\setminus\{i_k\})] + [u(S\setminus\{i_k\}) - u(S\setminus\{i_{k-1},i_k\})] + \cdots + [u(\{i_1\})-u(\emptyset)]\nonumber\\
			&\leq & u(\{i_k\}) + u(\{i_{k-1}\}) + \cdots + u(\{i_1\})\nonumber\\
			&=& \sum_{i\in S} \E[X_i]\nonumber\\
			&\leq & \sum_{i=1}^k\E[X_i]\label{equ:ineq1}
			\end{eqnarray}
			where the first inequality follows by the submodularity of function $u$, the second inequality is by the assumption that items are enumerated in decreasing order of their mean test scores. 
			
			By Jensen's inequality, for every $(x_1, x_2, \ldots,x_k) \in \reals_+^k$, we have
			$$
			\frac{1}{k}\sum_{i=1}^k x_i = \frac{1}{k}\sum_{i=1}^k (x_i^r)^{1/r} \leq \left(\frac{1}{k}\sum_{i=1}^k x_i^r\right)^{1/r}.
			$$
Hence, we have
\begin{equation}
			\sum_{i=1}^k\E[X_i] \leq k^{1-1/r} \E\left[\left(\sum_{i=1}^k X_i^r\right)^{1/r}\right].
\label{equ:ineqstar}
\end{equation} 
			From (\ref{equ:ineq1}) and (\ref{equ:ineqstar}), for every $S\subseteq N$ such that $|S| = k$,
			$$
			u(M) = \E\left[\left(\sum_{i=1}^k X_i^r\right)^{1/r}\right] \geq \frac{1}{k^{1-1/r}} \E\left[\left(\sum_{i\in S} X_i^r\right)^{1/r}\right] = \frac{1}{k^{1-1/r}}u(S).
			$$
			
			The tightness can be established as follows. Let $N$ consist of two disjoint subsets of items $M$ and $R$, where $M$ is a set of $k$ items whose each individual performance is of value $1+\epsilon$ with probability $1$, for parameter $\epsilon > 0$, and $R$ is a set of $k$ items whose each individual performance is of value $a$ with probability $1/a$ and of value $0$ otherwise, for parameter $a \geq 1$. Then, we note that
			$$
			u(M) = k^{1/r}(1+\epsilon)
			$$ 
			and
			\begin{eqnarray*}
				u(\OPT)  \geq  u(R) &=& \E\left[\left(\sum_{i\in R}X_i^r\right)^{1/r}\right] \\
				& \geq & a\Pr\left[\sum_{i\in R}X_i > 0\right]\\
				& = & a \left(1-\left(1-\frac{1}{a}\right)^k\right)\\
				& \geq & a \left(1-e^{-k/a}\right).
			\end{eqnarray*}
			
			Hence, it follows that
			$$
			\frac{u(M)}{u(\OPT)} \leq (1+\epsilon)\frac{1}{k^{1-1/r}} \frac{k/a}{1-e^{-k/a}}.
			$$
			
			The tightness claim follows by taking $a$ such that $k = o(a)$, in which case $(k/a)/(1-e^{-k/a})=1 + o(1)$.

\section{Proof of Proposition~\ref{thm:quant}}
\label{sec:quant}

		
\paragraph{Proof of Claim~(a)} If $k$ is a constant, then there is no $r$ satisfying both conditions $r=o(1)$ and $r>1$. Hence, it suffices to consider $k = \omega(1)$ and show that the following statement holds: for any given $\theta > 0$, there exists an instance for which greedy selection in decreasing order of quantile test scores cannot give a constant-factor approximation. 
			
			Consider the distributions of random variables $X_i$ defined as follows:
			\begin{enumerate}
				\item Let $X_i$ be equal to $a$ with probability $1$ for $1\le i \le k$. For each of these items, the quantile test score is equal to $a$ and the replication score is equal to $ak^{1/r}$.
				\item Let $X_i$ be equal to $0$ with probability $1-1/n$, and equal to $b \theta n/k$ with probability $1/n$ for $k+1 \le i \le 2k$. Note that in the limit as $n$ grows large, each of these items has quantile test score of value $b$ and replication score of value $b\theta$. 
				\item Let $X_i$ be equal to $0$ with probability $1-\theta/k$ and equal to $c$ with probability $\theta/k$ for $2k+1 \le i \le 3k$. For each of these items, the quantile test score is equal to $c$ and the replication test score is less than or equal to $c\theta^{1/r}$.
				\item Let $X_i$ be equal to $0$ for $3k+1 \le i \le n$.
			\end{enumerate}
			
			If $\theta$ is a constant, i.e., $\theta= O(1)$, we can easily check that greedy selection in decreasing order of quantile test scores cannot give a constant-factor approximation with $a=b=1$ and $c=2$. Under this condition, the selected set of items is $\{2k+1,\dots,3k\}$. However, we have 
			\begin{eqnarray*}
				\frac{\E\left[\left(\sum_{i=2k+1}^{3k} X_i^r \right)^{1/r} \right]}{\E\left[\left(\sum_{i=1}^{k} X_i^r \right)^{1/r} \right]} 
				&=&\frac{\E\left[\left(\sum_{i=2k+1}^{3k} X_i^r \right)^{1/r} \right]}{k^{1/r}}\cr
				&\le & \frac{\left(\sum_{i=2k+1}^{3k}\E\left[ X_i^r \right] \right)^{1/r}}{k^{1/r}} \cr
				& = & 2 \left( \frac{\theta}{k} \right)^{1/r} \\
				&= &~ o(1),
			\end{eqnarray*}
			which is because $k=\omega(1)$, $\theta = O(1)$, and $r= o(\log(k))$.
			
			Since $r > 1$, if $\theta$ goes to infinity as $n$ goes to infinity, i.e. for $\theta = \omega(1)$, we have 
			\begin{eqnarray*}
				\frac{\E\left[\left(\sum_{i=2k+1}^{3k} X_i^r \right)^{1/r} \right]}{\E\left[\left(\sum_{i=k+1}^{2k} X_i^r \right)^{1/r} \right]} 
				&\le & \frac{\left(\sum_{i=2k+1}^{3k}\E\left[ X_i^r \right] \right)^{1/r}}{\theta} \cr
				& = & 2 \theta^{(1-r)/r}\\
				& = & o(1).
			\end{eqnarray*}

Therefore, the greedy selection in decreasing order of quantile test scores has a vanishing utility compared to the optimal value. 
			
\paragraph{Proof of Claim~(b)} Let $T(X,S)$ be a subset of $S$ such that $i\in T(X,S)$ if, and only if, $X_i\geq P_i^{-1}(1-1/k)$, for $i\in S$. Let $a_{\max} = \max_{i \in S} a_{i}$ and $a_{\min} = \min_{i\in S} a_{i}$. We will show that there exist constants $\chigh$ and $\clow$ such that 
			$$ 
			\clow a_{\min} \le \E\left[\left(\sum_{i\in S} X_i^r \right)^{1/r} \right] \le \chigh a_{\max}. 
			$$
			
			Since $(x+y)^{1/r} \le x^{1/r} + y^{1/r}$ for all $x,y \ge 0$ and $r>1$, we have
			\begin{eqnarray*}
				\E \left[ \left(\sum_{i\in S} X_i^r \right)^{1/r} \right] & =  &
				\E\left[\left(\sum_{i\in T(X,S)} X_i^r +\sum_{i\in S\setminus T(X,S)} X_i^r \right)^{1/r} \right] \cr
				& \le & \E\left[\left(\sum_{i\in T(X,S)} X_i^r\right)^{1/r} +\left(\sum_{i\in S\setminus T(X,S)} X_i^r \right)^{1/r} \right] \cr
				& \le & \E\left[\sum_{i\in T(X,S)} X_i +\left(\sum_{i\in S\setminus T(X,S)} X_i^r \right)^{1/r} \right] \cr
				& \le & \E\left[\sum_{i\in T(X,S)} X_i +\left(\sum_{i\in S\setminus T(X,S)} (a_{\max})^r \right)^{1/r} \right] \cr
				&\le & \left( \E\left[ |T(X,S)| \right] + k^{1/r}  \right)a_{\max}\\ 
				&= & (1+k^{1/r}) a_{\max}.
			\end{eqnarray*}
			
			By the Minkowski inequality, $\left(\sum_{i \in A} \E \left[X_i \right] ^p\right)^{1/p}  \le \E \left[\left(\sum_{i \in A} X_i^p\right)^{1/p} \right] $ for all $A \subseteq S$. Thus, we have
			\begin{eqnarray*}
				\E\left[\left(\sum_{i\in S} X_i^p \right)^{1/p} \right] & =  &
				\E\left[\left(\sum_{i\in T(X,S)} X_i^p +\sum_{i\in S\setminus T(X,S)} X_i^p \right)^{1/p} \right] \cr
				&\ge & \E\left[\left(\sum_{i\in T(X,S)} X_i^p  \right)^{1/p}\right] \cr
				&=&\sum_{A\subseteq S} \Pr\{T(X,S) = A \}\E\left[\left(\left. \sum_{i\in A} X_i^p\right)^{1/p} \right| T(X,S) = A \right]  \cr
				&\ge&\sum_{A\subseteq S} \Pr[T(X,S) = A]\left(\sum_{i\in A} \E[X_i | i \in T(X,S) ]^p\right)^{1/p}  \cr
				&\ge & \sum_{A\subseteq S} \Pr\{T(X,S) = A \} |A|^{1/p} a_{\min}\cr
				&\ge & \left( 1- (1-1/k)^k\right) a_{\min}\\
				 & \ge & (1-1/e) a_{\min}.
			\end{eqnarray*}
			
Therefore, the greedy selection in decreasing order of quantile test scores gives a constant-factor approximation of the optimal value.

\end{APPENDIX}
%
%







\end{document}